\documentclass[12pt]{amsart}
\usepackage[margin = 1in]{geometry}
\usepackage{graphicx}
\usepackage{subfig}
\begin{document}

\newtheorem{theorem}{Theorem}
\newtheorem{lemma}{Lemma}
\newtheorem{asmp}{Assumption}

\title[Valuation using Volatility and Earnings]{Valuation Measure of the Stock Market using Stochastic Volatility and Stock Earnings}

\author{Andrey Sarantsev}

\address{Department of Mathematics and Statistics, University of Nevada, Reno}

\email{asarantsev@unr.edu} 

\author{Angel Piotrowski}

\author{Ian Anderson}

\begin{abstract}
We create a time series model for annual returns of three asset classes: the USA Standard \& Poor (S\&P) stock index, the international stock index, and the USA Bank of America investment-grade corporate bond index. Using this, we made an online financial app simulating wealth process. This includes options for regular withdrawals and contributions. Four factors are: S\&P volatility and earnings, corporate BAA rate, and long-short Treasury bond spread. Our valuation measure is an improvement of Shiller's cyclically adjusted price-earnings ratio. We use classic linear regression models, and make residuals white noise by dividing by annual volatility. We use multivariate kernel density estimation for residuals. We state and prove long-term stability results. 
\end{abstract}

\subjclass[2020]{62P05, 91G15, 91G70. Keywords: autoregression, stationarity, white noise, stochastic volatility, price-earnings ratio, bond duration, multiple linear regression. Corresponding author: A.S. \texttt{asarantsev@unr.edu} ORCID: 0000-0002-7415-4892}

\maketitle

\thispagestyle{empty}

\section{Introduction}

\subsection{Index levels and earnings} A classic way to find whether the stock market is priced fairly is to compare index levels with earnings. The classic {\it cyclically adjusted price-earnings} (CAPE) ratio, introduced by Robert Shiller, measures how reasonably the stock market (measured by the Standard \& Poor 500 and its predecessors) is priced relative to fundamentals (earnings). It is computed by dividing the current index level by annual earnings, averaged over the previous 10 years. Thie latter, called {\it cyclical averaging}, is done to smooth out wild short-term fluctuations in earnings, and take the data from both recessions and expansions. Historically, this ratio is between 5 and 45. Its long-term median value is 17. If it is much larger than that, then the market is overpriced and is likely to crash soon, or at least not grow as fast as before. See the classic article \cite{ShillerPaper} or the monograph \cite{ShillerBook}, or a more recent article \cite{Ural} and a textbook \cite[Chapter 8, Sections 4, 6]{Ilmanen}.

Robert Shiller did a remarkable job compiling historical data from 1871 for Standard \& Poor 500 index and its predecessors. Three main examples of such overprived markets in the USA history are: late 1920s followed by the Great Depression; the boom of 1960s followed by the stagflation and energy crises of 1970s; and the dotcom bubble in the late 1990s, followed by another stock crash and a lost decade of the 2000s. 

However, a surprising observation was made based on recent data: In the late 2010s and now (2025 as of this writing), the CAPE ratio is very high compared to historical average: Now, it is around 40. In fact, it is close to its all-time peak in the late 1990s at the height of the dotcom bubble. However, the market crash failed to materialize so far. There was a short-timed crash in early 2020 during the lockdowns, but the market has recovered quickly. Did the CAPE lose its predictive value? 

\subsection{The new valuation measure} In our previous article \cite{Valuation}, we modified the CAPE. We think that CAPE lost its importance because dividend payout decreased. Historically, about 70\% of earnings of stocks were paid as dividends to shareholders. However, in the last few decades, this decreased to 40\% on average. Instead of paying earnings as dividends, executives use earnings for stock buybacks. This increases the stock value. 

Reasons are still disputed in the academic community, but likely include tax issues, \cite[Chapter 8, Section 4]{Ilmanen}, \cite[Chapter 12, Section 5]{Ang}, as well as incentives for executives who get paid in bonuses of stocks in their own company, \cite[Chapter 15]{Ang}. During stock buybacks, overall earnings remain the same, but index levels increases. This boosts the CAPE ration, but does not make the market overpriced. 

To capture this, in \cite{Valuation} we compared annual total returns with annual earnings growth. In our research, annual returns are {\it geometric} (not arithmetic, which are commonly quoted); {\it total} (including dividends); and {\it nominal} (not adjusted for inflation).  Annual earnings growth is measured as logarithmic change, and also nominal (without inflation adjustments). The long-term difference is 4--5\% per year. Cumulative total returns minus earnings have linear trend with slope 4--5\%. After detrending, this becomes mean-reverting: Modeled as a classic autoregression of order 1. This detrended quantity is our new valuation measure. Our research showed significant predictive power for next year's total returns. 

\begin{figure}[t]
\subfloat[CAPE vs New Measure]{\includegraphics[width = 8cm]{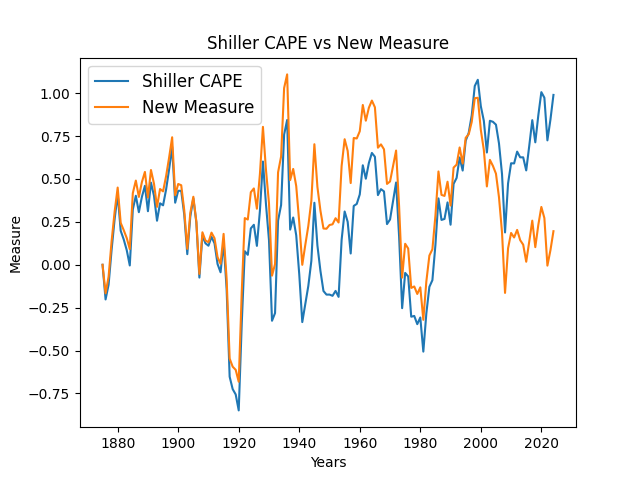}}
\subfloat[Volatility]{\includegraphics[width = 8cm]{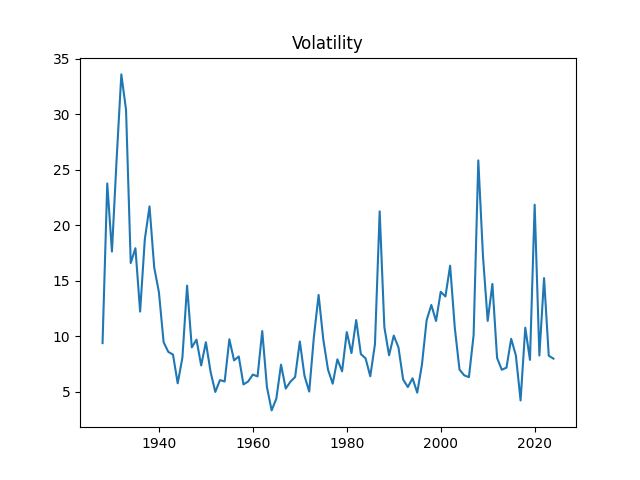}}
\caption{On the left: Shiller cyclically-adjusted price-earnings (CAPE) ratio on the log scale versus the new valuation measure proposed in \cite{Valuation}. They track each other before 2000, but after that they diverge. On the right: annual realized volatility for the Standard \& Poor. }
\label{fig:compare}
\end{figure}

This new valuation measure more or less tracks the CAPE in the 19th and 20th centuries. The new valuation measure reproduces these 3 peaks in the 1920s, 1960s, and 1990s mentioned above. However, after the dotcom bubble crashed in the 21st century, these two measures diverge. The CAPE shoots upward during the 2010s, but the new measure stays close to its historical average. This is a reason to believe that the stock market now, in fact, is reasonably priced, and there is no bubble.  Comparison of the Shiller CAPE and the new valuation measure is in \textsc{Figure}~\ref{fig:compare} (A).

\subsection{Market volatility} Assume we model earnings growth separately. Then we can model earnings and the index together, using this valuation measure. Indeed, this measure is the difference between total returns and earnings growth. Knowledge of earnings growth gives us knowledge of total returns. But it is much harder to model earnings growth terms: These are not independent identically distributed (IID), that is, not a (strong) white noise, judging by the autocorrelation function (ACF). In fact, the returns $Q(t)$ themselves have ACF close to zero, but the squared returns $Q^2(t)$ or their absolute values $|Q(t)|$ do not have zero ACF. For background, see the classic references \cite[Chapter 7]{Series}, especially Figures 7.1--7.3. A solution comes from using the volatility of the S\&P index. For background on such {\it stochastic volatility} models, see \cite[Chapter 7, Section 4]{Series} or \cite[Chapter 3, Section 1]{Textbook} (especially Figures 3.1, 3.2). 

We used monthly average volatility index (VIX), computed daily by the Chicago Board of Options Exchange (CBOE) from {\it implied volatility} (computed using prices of European options traded on this exchange). See background in \cite[Chapter 7, Section 2, Figure 7.3]{Ang}; also \cite[Chapter 10, Section 1]{Ang} for volatility-driven strategies; or \cite[Chapter 8, Section 4, Figure 8.5]{Ilmanen}. Dividing monthly S\&P 500 returns by this VIX made these returns IID Gaussian. Unfortunately, applying this methodology directly to our setting is problematic: We have annual (not monthly) data, and the VIX data goes back only as far as 1990 (or 1986 for the version of VIX for S\&P 100, a close relative of S\&P 500). 

To resolve this difficulty, we switched to annual {\it realized volatility} (as opposed to implied volatility): the standard deviation of daily index fluctuations (log changes) for S\&P 500 and its predecessor S\&P 90, within each year 1928--2024. This range of years is more limited than 1871--2024 from Robert Shiller's data library. Still, it includes the three great bubbles. See the graph on \textsc{Figure}~\ref{fig:compare} (B). 

What is more, we improved data for S\&P 500 itself: We take end-of-year (December 31 or the last working day of each year) index level (available from 1927) instead of the average January daily close data (available from Robert Shiller's data library from 1871). This is important, since monthly average is not directly investable, but the end-of-year level is. 

Annual returns are not IID, but dividing them by annual volatility makes them IID. Same for the annual earnings growth. The volatility itself can be modeled as an autoregression of order 1, if you take the logarithm of volatility. Thus we described two factors influencing the stock market returns: Volatility and the new valuation measure (computed using earnings). 

\subsection{Bond rates and spreads} Other factors influence the S\&P 500 total returns. Bond and stock markets compete for the capital of investors. An increase in bond rates makes the bond market more attractive, thus reducing demand for stocks. As a benchmark, we took the BAA rate for US corporate bonds. This Moody's rating is less than the perfect AAA, and is on the lower range for investment-grade bonds. Still, these bonds are relatively safe from defaults. We take average daily rate in December of 1927--2024. See the classic article \cite{Bernanke} for reference. For background, see textbooks \cite[Chapter 9, Section 4]{Ang} and \cite[Chapter 9, Section 4]{Ilmanen}. 

\begin{figure}[t]
\subfloat[Spread]{\includegraphics[width = 8cm]{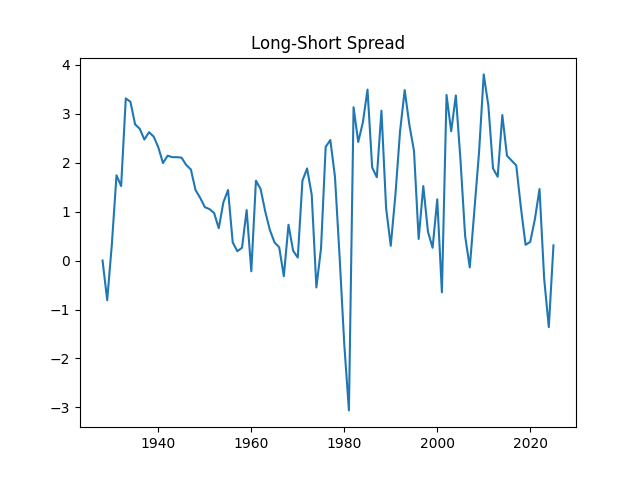}}
\subfloat[BAA rate]{\includegraphics[width = 8cm]{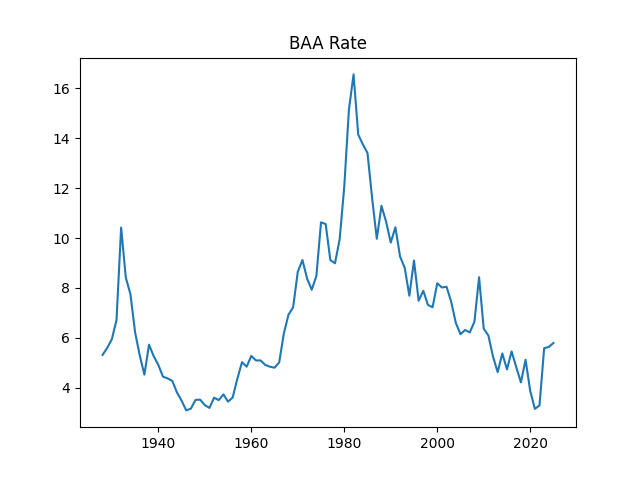}}
\caption{The long-short (10 year minus 3 month) Treasury spread, and the Moody's BAA rate: Average daily December data.}
\label{fig:bonds}
\end{figure}

Yet another factor is the long-short Treasury spread, sometimes called {\it term spread}. Usually, long-term Treasury bonds have higher rate than short-term Treasury bills. This is natural, since committing money for a longer time is more risky and demands compensation in the form of a higher rate. Occasionally, however, this comparison inverts: Long-term Treasuries have lower rate than short-term Treasuries. This phenomenon is called the {\it inverted yield curve}. This usually happens when the market anticipates an impeding recession. 

The Federal Reserve System has direct influence over short-term rates via its standard {\it open-market operations}. When a recession strikes, the Federal Reserve lowers short-term rates. However, the long-term rate now combines the current short-term rate and the anticipated future short-term rates. Indeed, an alternative to buying long-term bonds is reinvesting many times in short-term bonds. In sum, when long-term rates are lower than the short-term rates, then a recession and a stock market crash are likely to happen soon. For the term spread, we consider December average daily data for 10 years minus 3 month Treasury rates. For reference, see for example \cite{FF1989}, and more recent textbooks \cite[Chapter 9, Section 3]{Ang}, especially Figure 9.9; or \cite[Chapter 9]{Ilmanen}. 

We model both BAA rate and the term spread as autoregression of order 1. For the BAA rate, we take its logarithm to prevent it from being negative, similarly to the volatility.  See the graphs in~\textsc{Figure}~\ref{fig:bonds}. 

\subsection{Our main contributions} In our previous research, we proved long-term stability only for the valuation measure. Here, we have a feedback loop: The stock market returns are modeled using linear regression with factors including the new valuation measure, but the latter in return is computed using earnings growth and the stock market returns. Thus we need to state and prove a non-trivial long-term stationarity result. It does not follow immediately from our assumptions. 

Also, the simulator includes the international stocks and USA bonds. This is why we model this using linear regression as well, using rates and volatility. We provide full statistical analysis and prove the stationarity for this extended model. We apply this simulator for retirement modeling, to test the classic 4\% withdrawal rule.

The Excel data and Pyhon code for the fitting of the model, analysis of residuals, and the web app itself, is available in the repository \texttt{asarantsev/sim} Web app has back-end written in Python, using the Flask framework, with HTML for the front-end. The web address is \texttt{asarantsev.pythonanywhere.com} 

\subsection{Organization of the article} In Section 2, we describe the raw data, and the preliminary processing. We discuss the sources and forats of data in detail. Section 3 is devoted to methodology of white noise analysis. In Section 4, we describe the simple model with only two factors: volatility and BAA rate. In Section 5, we describe in detail the new valuation measure, and use it for combined model of S\&P returns and earnings. In Section 6, we consider the complete model. For each model, we state and prove the main stability result. Section 7 describes the online financial simulator based on the model from Section 6. 

\section{Data Description} 

\subsection{Preliminary remarks} For this background, see \cite[Chapter 1, Section 1]{Textbook}. We consider total returns (including dividends or coupon payments, not just price changes), nominal (not inflation-adjusted) rather than real (inflation-adjusted), and geometric (log returns) rather than arithmetic: If 1\$ turned into 2\$ the returns are $\ln 2$ not 100\%. Arithmetic and geometric returns $r_a, r_g$ are related as $r_g = \ln(1 + r_a)$. However, if $r_a, r_g$ are small, then $r_a \approx r_g$. 

Using geometric returns allows us to aviod compound interest: If $r_1$ and $r_2$ are geometric returns of this asset class in years 1 and 2, then $r_1 + r_2$ are geometric returns of this asset class in two years. However, arithmetic returns are better for computing portfolio returns: If $r_1$ and $r_2$ are arithmetic returns of asset classes 1 and 2 in year $t$, and if $\pi_1$ and $\pi_2$ are portfolio weights for a constant-weighted portfolio, then $\pi_1r_1 + \pi_2r_2$ are arithmetic returns of this portfolio in year $t$.

\subsection{Online data resources} We used the following web sites:

\begin{itemize}
\item Robert Shiller's Yale University financial data library (Shiller): \texttt{shillerdata.com}
\item Federal Reserve Economic Data (FRED): \texttt{fred.stlouisfed.org}
\item Yahoo Finance: \texttt{finance.yahoo.com}
\item Novel Investor web site: \texttt{novelinvestor.com/historical-returns}
\end{itemize}

\subsection{Raw data} All data is nominal.

\begin{itemize}
\item Daily close December 30, 1927--December 31, 2024 for Standard \& Poor (S\&P): Standard \& Poor 500, 1957--2024, and its predecessor, Standard \& Poor 90, 1927--1956. {\it Source:} Yahoo Finance \^{}SPX series
\item Annual S\&P earnings $E(t)$, 1918--2024, and dividends, 1928--2024. {\it Source:} Shiller, file \texttt{ie\_data.xls} columns named Dividend D and Earnings E.
\item Moody's BAA corporate bond rate $R(t)$, average daily close December, 1927--2024. {\it Source:} FRED, BAA series
\item Treasury long rates, average daily close December, 1927--2024. {\it Source:} FRED series DGS10, 10-year Treasury rate
\item Treasury short rates, average daily close December, 1927--2024. {\it Source:} FRED series TB3MS, 3-month Treasury rate, 1934--2024, FRED series M1329AUSM193NNBR, 3- and 6-month short-term certificates of deposit, 1927--2024
\item Bank of America (BofA ICE) Corporate Bonds Total Return Index Value, $\mathcal U(t)$  end-of-year for 1972--2024. {\it Source:} FRED series BAMLCC0A0CMTRIV
\item International annual total returns: Morgan Stanley Capital International (MSCI) Europe, Australasia \& Far East (EAFE), 1970--2024, index; MSCI Emerging Markets (EM), 1988--2024, index. {\it Source:} Novel Investor
\end{itemize}

Using this raw data, we compute the following:

\begin{itemize}
\item From the daily S\&P index, we compute daily log changes of S\&P. Within each year 1928--2024, we compute the standard deviation of these daily log changes. This is the annual S\&P (realized) {\it volatility} $V(t)$, 1928--2024.
\item We subtract Treasury short rates from Treasury long rates, and get Treasury {\it term spread} $S(t)$, or {\it long-short spread}, 1927--2024.
\item Total S\&P annual returns $Q(t)$ are computed from end-of-year S\&P index and its annual dividends, 1928--2024.
\item For the averaging window $L = 1, \ldots, 10$, trailing averaged earnings $\overline{E}(t)$ of year $t$ are the average of $E(t), \ldots, E(t-L+1)$ annual earnings, for 1927--2024.
\item We compute earnings growth $G(t)$ as log changes in annual earnings, 1928--2024, and averaged earnings growth $\overline{G}(t)$ as log changes in annual averaged earnings, 1928--2024.
\item We compute international stock returns $I(t)$ as EAFE from 1970--1987 and the 60/40 portfolio of EAFE/EM from 1988--2024, thus having the data 1970--2024. The 60/40 choice is arbitrary and can be changed in further research.
\item Compute corporate bond returns $B(t)$ as log changes in their Index Value, 1928--2024. 
\end{itemize}

We also provide the data on end-of-year annual S\&P index, earnings, dividends, bond rates, and $I(t), \mathcal U(t)$ from above, on our web site \texttt{asarantsev.github.io} and in GitHub repository \texttt{github.com/asarantsev/sim}

\subsection{Asset classes and return factors} Let us describe in detail the three asset classes. We denote their total returns by letters. 

\begin{itemize}
\item USA stocks $Q$, measured by the S\&P index, 1928--2024.
\item International stocks $I$, measured by the MSCI indices: The EAFE index, 1970--1987, and the 60/40 portfolio of the EAFE/EM indices, 1988--2024.
\item USA investment-grade corporate bonds $B$, measured by BofA ICE Corporate Bond Index, 1973--2024.
\end{itemize}

The four factors which we use for modeling in this article:

\begin{itemize}
\item Annual S\&P volatility $V$
\item Moody's BAA corporate bond rate $R$
\item Long-short term spread $S$
\item Annual S\&P earnings $E$
\end{itemize}

\section{White Noise Analysis Methodology}

The code is \texttt{Monte-Carlo.py} and results are in~\textsc{Table}~\ref{table:values}. We shall use Shapiro-Wilk and Jarque-Bera normality tests, as well as skewness, kurtosis, and the ACF. Using Monte Carlo simulation, we provide critical values for IID $x_1, x_2, \ldots, x_N \sim \mathcal N(\mu, \sigma^2)$ for the following statistics: empirical skewness, empirical kurtosis (with both  theoretical skewness and kurtosis equal zero for normal distribution; and the $L_1$ norm of the empirical ACF $\hat{\rho}(k)$ with lag $k$. Here, $\overline{x}$ and $s$ are empirical mean and standard deviation of the data $x_1, \ldots, x_N$.  
$$
L_1 = \sum_{k=1}^5|\hat{\rho}(k)|,\quad \hat{\rho}(k) = \frac1{Ns^2}\sum\limits_{j=1}^{N-k}(x_j - \overline{x})(x_{j+k} - \overline{x}).
$$
For background on white noise and normality tests, see \cite[Chapter 1, Section 6]{Series}. Also, for white noise testing, see \cite[Chapter 1, Sections 2, 4]{Textbook},  and for normality testing, see \cite[Chapter 1, Section 5]{Textbook}; finally, see \cite[Chapter 2, Section 6]{Textbook}, and the result (2.10) there. 

\begin{table}[h]
\centering
\begin{tabular}{|ccccc|}
\hline
Size $N$ & $p$ & Skew & Kurt & L1 \\
\hline
100 & 95\% & 0.47 & 0.76 & 0.63 \\
100 & 99\% & 0.63 & 1.38 & 0.76 \\
50 & 95\% & 0.64 & 0.98 & 0.88 \\
50 & 99\% & 0.88 & 1.86 & 1.08 \\
\hline
\end{tabular}
\caption{Critical values for skewness, kurtosis, and $L_1$ norm of the first 5 lags of the ACF for original and absolute values of standard normal variables.}
\label{table:values}
\end{table}

Note that we perform the ACF computation both for $x_1, \ldots, x_N$ and for $|x_1|, \ldots, |x_N|$. The necessity of doing this not just for $x_1, \ldots, x_N$ is shown in  \cite[Chapter 3, Section 1]{Textbook} (especially Figures 3.1, 3.2), mentioned in the Introduction.  But both $x_1, \ldots, x_N$ (the original values) and $|x_1|, \ldots, |x_N|$ (absolute values) give us the same critical values of the $L_1$ statistic.

We recall the standard time series results: Under the white noise hypothesis, we have weak convergence:
$(\hat{\rho}(1), \ldots, \hat{\rho}(k)) \Rightarrow \mathcal Q_k$, where $N \to \infty$, where $\mathcal Q_k$ is the standard $k$-variate normal distribution. From here, it is easy to derive asymptotic results for $L_1$ and $L_2$. But we decided to make Monte Carlo simulations and compute critical values by hand, since we are not sure about the rate of convergence. After all, this was in part what motivated the replacement of the standard $L_2$ Box-Pierce test with the Ljung-Box test, see the classic article \cite{Ljung}.

\section{The Simplified Model}

In this section, we consider the model of returns of three asset classes, but using only two of the four above factors: S\&P volatility $V$ for stocks and rate $R$ for bonds. This was similar to the model implemented in previous versions of the simulator. We discuss this model in detail to illustrate our main ideas. The Python code is \texttt{simple.py} 

\subsection{Stock returns} It is reasonable to model domestic stock returns $Q$ by terms which are IID but not Gaussian. But after normalizing (dividing by volatility), $Q/V$ can be modeled as IID Gaussian, see \textsc{Figure}~\ref{fig:basic-returns}. See also \textsc{Table}~\ref{table:simple-normalizing}. Suprpisingly, we see it is reasonable to model international returns $I$ as IID Gaussian even without normalizing. But we think it is a statistical illusion because we do not have that much annual data for international stocks. We still decide to normalize international returns and consider $I/V$. 

It is remarkable that we can normalize international stock returns by dividing by domestic volatility. This works for international as well as for domestic stocks. We explain it in the following way: World equity markets are interconnected, and when volatility of domestic markets increases, the same happens with international markets.

\begin{table}[h]
\centering
\begin{tabular}{lrrrrrrrr}
\hline
returns & stdev & skew & kurt & SW & JB & L1O & L1A \\
\hline
$Q$  & 0.1853 & -0.971 &  1.082 & 0.000 & 0.000 & 0.327 & 0.265 \\
$Q/V$ & 0.0197 &  0.228 & -0.211 & 0.115 & 0.601 & 0.558 & 0.497 \\
$I$  & 0.2120 & -0.617 &  1.015 & 0.168 & 0.054 & 0.488 & 0.563 \\
$I/V$ & 0.0236 &  0.604 & -0.129 & 0.049 & 0.184 & 0.392 & 0.207 \\
\hline
\end{tabular}
\caption{Descriptive statistics of domestic and international stock returns before and after normalizing. Skew = skewness, Kurt = kurtosis, SW and JB =  Shapiro-Wilk and Jarque-Bera $p$-values, L1O and L1A = L1 norm for the first 5 lags of the ACF for $Z$ and $|Z|$}
\label{table:simple-normalizing}
\end{table}

\begin{figure}
\subfloat[$Q$]{\includegraphics[width = 5cm]{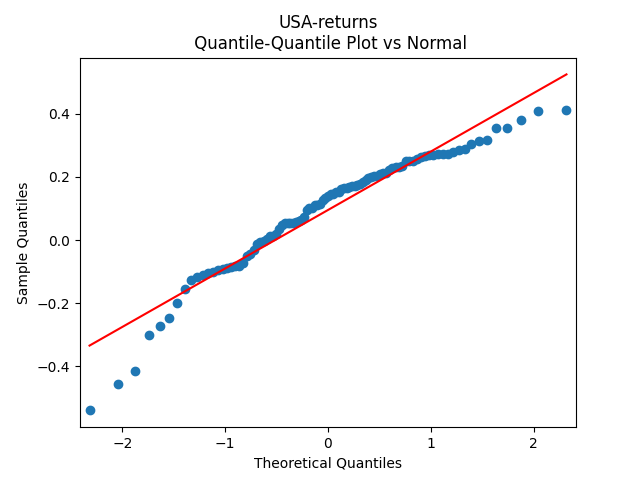}}
\subfloat[$Q$]{\includegraphics[width = 5cm]{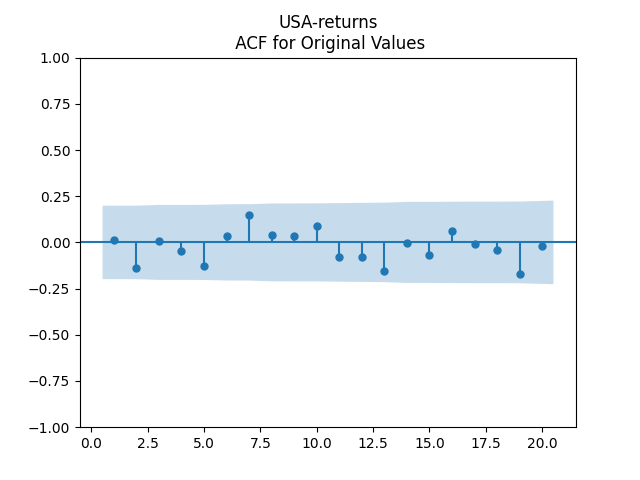}}
\subfloat[$|Q|$]{\includegraphics[width = 5cm]{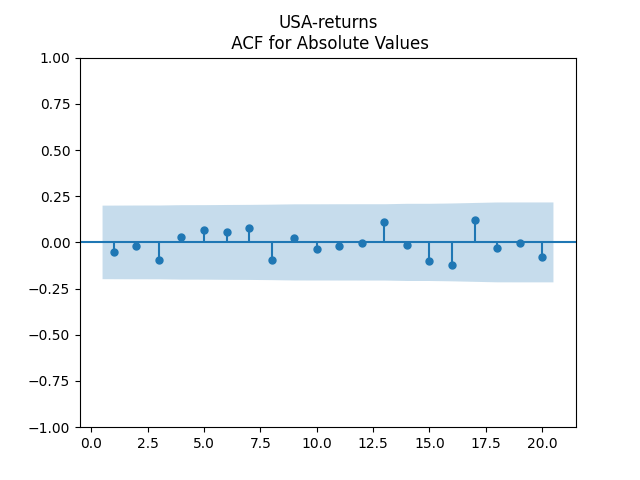}}
\newline
\subfloat[$Q/V$]{\includegraphics[width = 5cm]{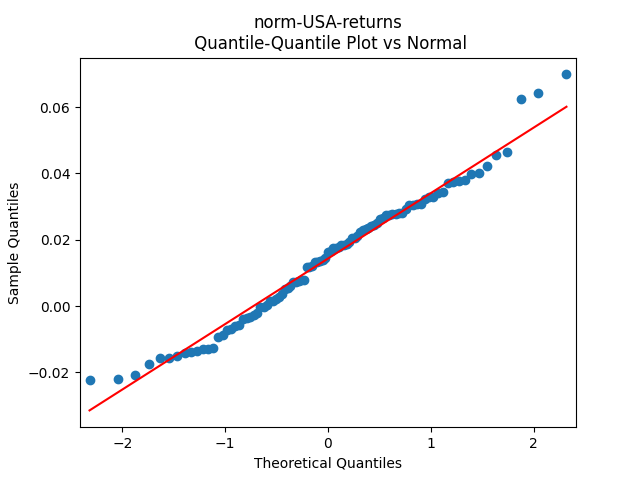}}
\subfloat[$Q/V$]{\includegraphics[width = 5cm]{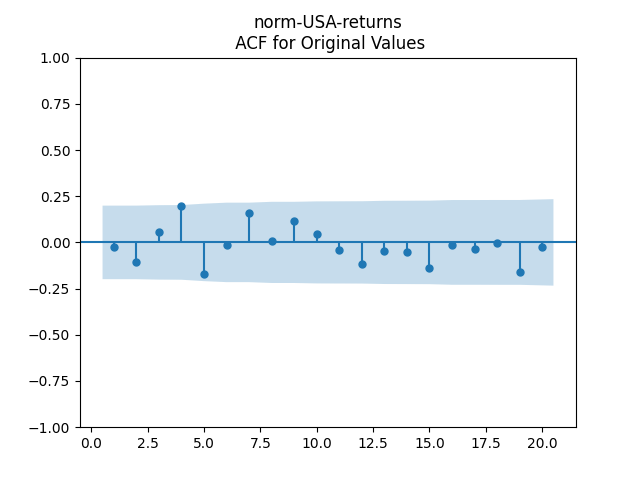}}
\subfloat[$|Q/V|$]{\includegraphics[width = 5cm]{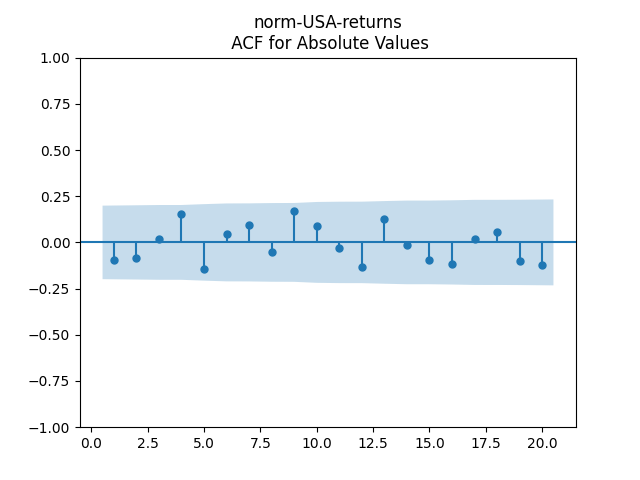}}
\caption{Top: The quantile-quantile plot for $Q$, and the ACF for $Q$ and for $|Q|$, where $Q$ is S\&P returns. This shows $Q$ are IID but not Gaussian. Bottom: The quantile-quantile plot for $Q/V$, and the ACF for $Q/V$ and for $|Q/V|$, where $Q$ is total S\&P returns. This shows $Q/V$ are IID Gaussian.}
\label{fig:basic-returns}
\end{figure}

\subsection{Stock market volatility} Following our previous manuscript with Jihyun Park with monthly implied volatility analysis, annual volatility on the log scale is modeled as an autoregression of order 1:
\begin{equation}
\label{eq:vol}
\ln V(t) = \alpha_V  + \beta_V\ln V(t-1) + Z_V(t).
\end{equation}
We use the log scale to make sure the volatility is always positive. Fitting~\eqref{eq:vol}, we get: $\alpha_V = 0.847850$ and $\beta_V = 0.620146$. We reject the hypothesis $\beta_V = 1$ (the random walk), with $p < 0.1\%$. Analysis of residuals shows that they are IID but not Gaussian, see \textsc{Table}~\ref{table:residuals-simple}.

\subsection{Bond rates and returns} Returns $B(t)$ of bonds are modeled using bond rates $R(t)$. The simplest approximation of bond returns during year $t$ is 
\begin{equation}
\label{eq:simplest-bond}
B(t) \approx 0.01\cdot R(t-1).
\end{equation}
This is true if we had arithmetic returns $B(t)$; but geometric returns are close to arithmetic returns if both are small. A bond with 5\% rate at end of year $t-1$ will return approximately 5\% in year $t$. However, all bonds have {\it interest rate risk:} Changes in rates during year $t$ affect the price of this bond. Sensitivity to this change is called {\it duration.} An increase in rates leads to a decrease in prices. This duration is approximately the average of dates for payments (coupon payments, usually every 6 months; and the principal at maturity), weighted by the size of these payments. For example, a {\it zero-coupon bond} has duration equal to its maturity. But corporate bonds usually do have coupons. Assuming $d_B$ is this duration, we add the following term to~\eqref{eq:simplest-bond}:
\begin{equation}
\label{eq:duration-bond}
B(t) \approx 0.01\cdot R(t-1) - d_B(R(t) - R(t-1)).
\end{equation}
Finally, corporate bonds do have a small risk of default: Not paying in full their coupons or the principal. This stands in contrast with Treasury bonds, which have interest rate risk but not default risk. We need to substract a constant from the right-hand side of~\eqref{eq:duration-bond}:
\begin{equation}
\label{eq:B}
B(t) - 0.01\cdot R(t-1) = -a_B - d_B(R(t) - R(t-1)) + Z_B(t).
\end{equation}
Results are: $a_B = 0.016611$ and $d_B = 0.055884$. The parameter $m_0$ stands for the {\it duration:} sensitivity of corporate bond prices to changes in interest rates. Here, the duration is about $5.5$ years. The parameter $a_B$ measures {\it default rates}. The $p$-values for the Student $T$-test are less than $0.1\%$.  Finally, the BAA corporate bond rate is also modeled as an autoregression of order 1, after taking the logarithm to ensure they stay positive:
\begin{equation}
\label{eq:baa}
\ln R(t) = \alpha_R + \beta_R\ln R(t-1) + Z_R(t).
\end{equation}
Fitting~\eqref{eq:baa}, we get:  $\alpha_R = 0.643694$ and $\beta_R = 0.539518$. Note that the $p$-value for $\beta_R = 1$ using Student $T$-test is $p = 9.4\%$. Thus we fail to reject the null hypothesis that $\beta_R = 1$, that is, the process $\ln R$ is a random walk.

\subsection{Combined simple model} We combine equations ~\eqref{eq:vol},~\eqref{eq:baa},~\eqref{eq:B} with the observations about normalized returns:

\begin{align}
\label{eq:simple-system}
\begin{split}
\frac{Q(t)}{V(t)} &= g_Q + Z_Q(t);\\
\frac{I(t)}{V(t)} &= g_I + Z_I(t);\\
\ln V(t) &= \alpha_V  + \beta_V\ln V(t-1) + Z_V(t);\\
\ln R(t) &= \alpha_R + \beta_R\ln R(t-1) + Z_R(t);\\
B(t) &= 0.01\cdot R(t-1) -a_B - d_B(R(t) - R(t-1)) + Z_B(t);\\
\end{split}
\end{align}

\subsection{Including duration for stocks} Bond returns are influenced by the duration, discussed earlier in~\eqref{eq:B}. We can also consider stock duration: Regress $Q(t)$ and $I(t)$ versus $R(t) - R(t-1)$. However, we do not wish to write a simple linear regression:
\begin{equation}
\label{eq:simple-linear}
Q(t) = c_Q - d_Q(R(t) - R(t-1)) + Z_Q(t),
\end{equation}
since we normalized $Q$ by dividing it by volatility. Only this makes data Gaussian. So we rewrite~\eqref{eq:simple-linear}
Instead of~\eqref{eq:simple-system}, we get: 
\begin{equation}
\label{eq:new-linear}
\frac{Q(t)}{V(t)} = \frac{c_Q}{V(t)} - d_Q\frac{R(t) - R(t-1)}{V(t)} + Z_Q(t).
\end{equation}
But usually, a linear regression has an intercept. We add such intercept to~\eqref{eq:new-linear} and get:
\begin{equation}
\label{eq:vol-linear}
\frac{Q(t)}{V(t)} = a_Q - d_Q\frac{R(t) - R(t-1)}{V(t)} + \frac{c_Q}{V(t)} + Z_Q(t).
\end{equation}
We can rewrite~\eqref{eq:vol-linear} as $Q(t) = a_QV(t) - d_Q(R(t) - R(t-1)) + c_Q + V(t)Z_Q(t)$. The same can be done for returns of international instead of domestic stocks. Combining all these, we get:

\begin{align}
\label{eq:duration-system}
\begin{split}
Q(t) &= a_QV(t) - d_Q(R(t) - R(t-1)) + c_Q + V(t)Z_Q(t);\\
I(t) &= a_IV(t) - d_I(R(t) - R(t-1)) + c_I + V(t)Z_I(t);\\
\ln V(t) &= \alpha_V  + \beta_V\ln V(t-1) + Z_V(t);\\
\ln R(t) &= \alpha_R + \beta_R\ln R(t-1) + Z_R(t);\\
B(t) &= 0.01\cdot R(t-1) -a_B - d_B(R(t) - R(t-1)) + Z_B(t);\\
\end{split}
\end{align}

We have coefficients
\begin{align*}
a_Q &= 0.2121,\quad d_Q = 0.0621,\quad c_Q = -0.0109,\\
a_I &= 0.2829,\quad d_I = 0.0508,\quad c_I = -0.0189.
\end{align*}
The coefficients $d_Q, c_Q, d_I, c_I$ have Student $T$-test $p$-values $<0.1\%, 0.8\%, 2\%, 2\%$. Analysis of residuals show that all five series of residuals are IID, but only $Z_Q, Z_I, Z_B$ are Gaussian. 

\begin{table}
\centering
\begin{tabular}{|ccccccccc|}
\hline
Reg & Size & Stdev & Skew & Kurt & SW & JB & L1O & L1A \\
\hline
$Z_V$ & 96 & 0.3644 & 0.590 & 0.057 & 0.009 & 0.061 & 0.401 & 0.237 \\
$Z_R$ & 97 & 0.1357 & 0.814 & 1.409 & 0.007 & 0.000 & 0.177 & 0.360  \\
$Z_Q$ & 97 & 0.0149 & 0.226 & 0.068 & 0.420 & 0.655 & 0.480 & 0.435 \\
$Z_I$ & 55 & 0.0183 & -0.059 & -0.501 & 0.754 & 0.738 & 0.401 & 0.519 \\
$Z_B$ & 52 & 0.0263 & 0.193 & 0.238 & 0.857 & 0.800 & 0.878 & 0.706 \\
\hline
\end{tabular}
\caption{Analysis of regression residuals. Skew = skewness, Kurt = kurtosis, SW and JB =  Shapiro-Wilk and Jarque-Bera $p$-values, L1O and L1A = L1 norm for the first 5 lags of the ACF for $Z$ and $|Z|$, Length = the size of $Z$}
\label{table:residuals-simple}
\end{table}

The empirical correlation matrix of $(Z_V, Z_R, Z_Q, Z_I, Z_B)$ is 
$$
\mathbf{\Sigma}_5 := 
\begin{bmatrix}
1 & -0.3 & -0.12 & -0.1 & 0.07\\
0.3 & 1 & -0.14 & -0.04 & -0.15\\
-0.12 & -0.14 & 1 & 0.29 & 0.41\\
-0.1 & -0.04 & 0.29 & 1 & -0.07\\
0.07 & -0.15 & 0.41 & -0.07 & 1\\
\end{bmatrix}
$$

\begin{theorem} If $\beta_V, \beta_R \in (0, 1)$, and if $(Z_Q(t), Z_I(t), Z_B(t), Z_V(t), Z_R(t))$ are IID with mean zero, the system~\eqref{eq:duration-system} has a unique stationary version. 
\end{theorem}

\section{Modeling Earnings Growth and the New Valuation Measure}

\subsection{The valuation measure} In this subsection, we mimic the econometric analysis of \cite[Section 2.1]{Valuation}. Define the wealth process as $W(t) = \exp(Q(1) + \ldots + Q(t))$, with initial wealth $W(0) = 1$. We use $L$-year averaged trailing earnings $\overline{E}$ for $L = 10$, and their growth:
\begin{equation}
\label{eq:trailing}
\overline{E}(t) = \frac1L\sum\limits_{s=0}^{L-1}E(t-s),\quad \overline{G} = \ln \overline{E}(t) - \ln\overline{E}(t-1).
\end{equation}
We compare wealth with averaged trailing earnings:
\begin{equation}
\label{eq:w-e}
\ln W(t) - \ln \overline{E}(t) = Q(1) + \ldots + Q(t) - \overline{G}(1) - \ldots - \overline{G}(t) - \ln E(0).
\end{equation}
As discussed in the Introduction, long-term growth of $\ln W$ is 6--7\% and long-term growth of $\ln E$ (and therefore $\ln\overline{E}$) is 2\%. Therefore, the long-term growth $c$ of such difference from~\eqref{eq:w-e} is 4--5\%. After detrending (subtracting the linear trend $ct$, which is to be determined), we get the {\it new valuation measure}
\begin{equation}
\label{eq:bubble-def}
H(t) = \ln W(t) - \ln\overline{E}(t) - ct = \sum\limits_{s=1}^t(Q(s) - \overline{G}(s)) - ct.
\end{equation}
We model this quantity $H(t)$ from~\eqref{eq:bubble-def} using autoregression of order 1:
\begin{align}
\label{eq:bubble-ar}
H(t) - h = b(H(t-1) - h) + U(t).
\end{align}
Here, $h$ plays the role of the long-term average of this new valuation measure $H$, and $b$ is the mean-reverting slope coefficient of this autoregression. 

\begin{figure}[h]
\includegraphics[width = 12cm]{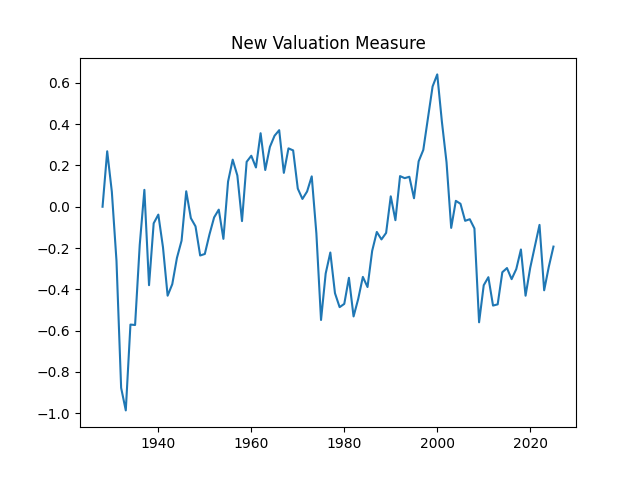}
\caption{The new valuation measure, computed in this article, using end-of-year S\&P data instead of January average S\&P data, and 10-year instead of 5-year averaging window. It closely tracks but is slightly different from the new valuation measure on \textsc{Figure}~\ref{fig:compare}. However, it shows the three main peaks in 1920s, 1960s, and 1990s. As of 2024, the stock market is not overvalued.}
\label{fig:bubble}
\end{figure}

Using~\eqref{eq:bubble-def}, we rewrite~\eqref{eq:bubble-ar} below to fit as an ordinary linear regression. 
\begin{equation}
\label{eq:bubble-reg}
Q(t) - \overline{G}(t) = \alpha +  \beta(t-1) - \gamma\sum\limits_{s=1}^{t-1}(Q(s) - \overline{G}(s)) + U(t).
\end{equation}

\begin{table}[h]
\centering
\begin{tabular}{|cccc|}
\hline
\bf{coefficient} & \bf{estimate} & \bf{Student} $p$-\bf{value} & \bf{stderr} \\
\hline
$\alpha$ & 0.023893 & 51.6\% & 0.037\\
$\beta$ & 0.008608 & 0.3\% & 0.003\\
$\gamma$ & 0.190133 & 0.2\% & 0.060\\
\hline
\end{tabular}
\caption{Regression results for~\eqref{eq:bubble-reg}. Coefficients $\beta$ and $\gamma$ are statistically significant. Innovations $U$ are IID but not Gaussian.}
\label{table:bubble-fit}
\end{table}

Results of fit of~\eqref{eq:bubble-reg} are shown in \textsc{Table}~\ref{table:bubble-fit}. Solving for $h, b, c$ in~\eqref{eq:bubble-def},~\eqref{eq:bubble-ar}, similarly to \cite[(6)]{Valuation}, we get:
\begin{equation}
\label{eq:bubble-result}
b = 1 - \gamma = 0.819867,\quad c = \frac{\beta}{\gamma} = 0.04527,\quad h = \frac{\alpha - c}{\gamma} = -0.1124.
\end{equation}
We reject the hypothesis $b = 1$, so the autoregression~\eqref{eq:bubble-ar} is mean-reverting, not a random walk. Residuals $U$ are IID but not Gaussian, judging by statistics in \textsc{Table}~\ref{table:residuals-bubble}.

\subsection{Earnings growth} For annual earnings $E(t)$, their growth $G(t) = \ln E(t) - \ln E(t-1)$ can be also modeled using annual volatility. Terms $G(t)$ are not IID. But dividing these growth terms by annual market volatility $V$ makes them white noise: $G(t)/V(t)$ are IID and Gaussian, see \textsc{Figure}~\ref{fig:growth}. To us, it is surprising that dividing $G$ by $V$ makes these IID. After all, $V$ is the volatility of daily index changes, not annual earnings. However, this method does work. The quantity $V$ plays universal role for stock markets: It improves (makes closer to IID) S\&P returns, international returns, and S\&P earnings. 

\begin{figure}
\subfloat[$G$]{\includegraphics[width = 5cm]{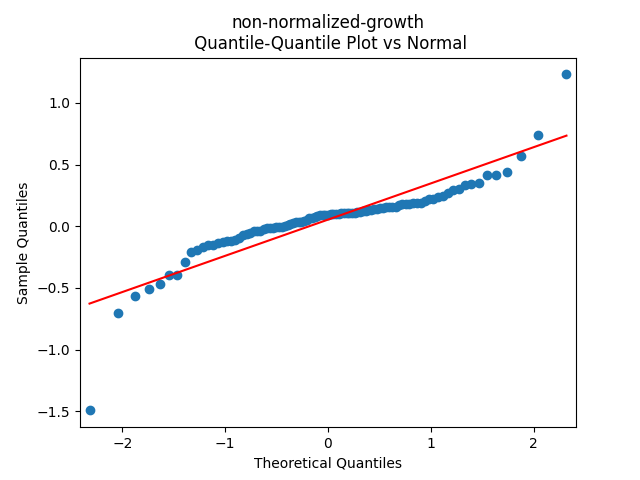}}
\subfloat[$G$]{\includegraphics[width = 5cm]{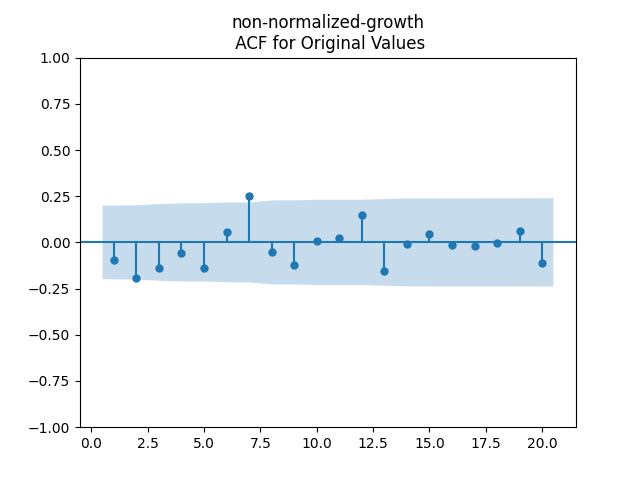}}
\subfloat[$|G|$]{\includegraphics[width = 5cm]{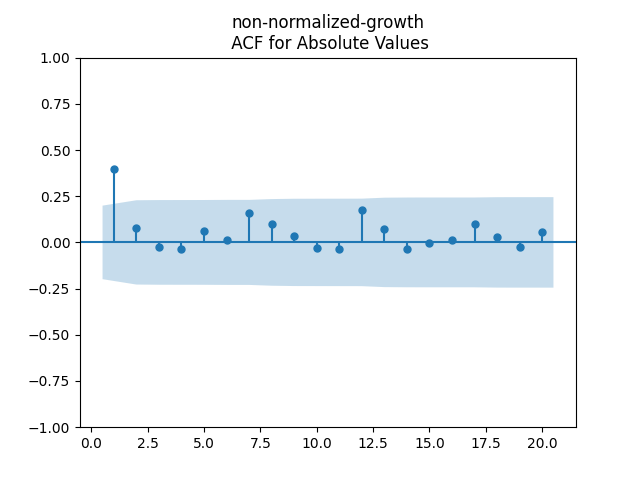}}
\newline
\subfloat[$G/V$]{\includegraphics[width = 5cm]{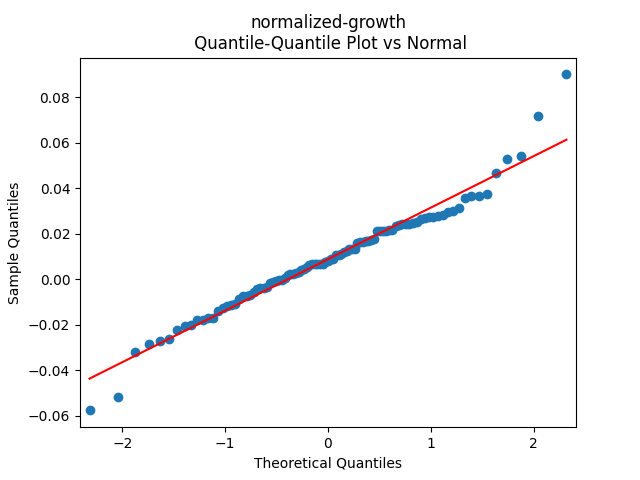}}
\subfloat[$G/V$]{\includegraphics[width = 5cm]{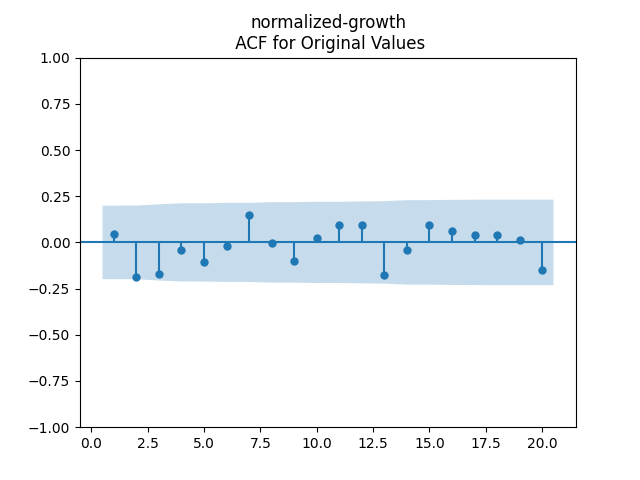}}
\subfloat[$|G/V|$]{\includegraphics[width = 5cm]{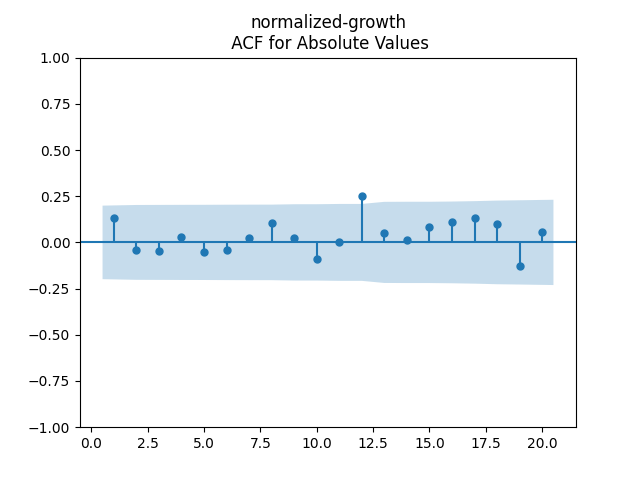}}
\caption{Top: The quantile-quantile plot for $G$, and the ACF for $G$ and for $|G|$, where $G$ is S\&P returns. This shows $G$ are not IID and not Gaussian. Bottom: The quantile-quantile plot for $G/V$, and the ACF for $G/V$ and for $|G/V|$, where $G$ is total S\&P returns. This shows $G/V$ are IID Gaussian.}
\label{fig:growth}
\end{figure}

\subsection{Combined model of S\&P returns and earnings} This gives us the following model for S\&P returns and earnings. This is an improvement of our previous article \cite{Valuation} where we did not model earnings growth. Therefore, our model in that article was not complete. Combine~\eqref{eq:w-e},~\eqref{eq:bubble-def},~\eqref{eq:bubble-ar} with definitions of $G$ and $\overline{G}$.

\begin{table}[ht]
\centering
\begin{tabular}{|ccccccccc|}
\hline
Reg & Size & Stdev & Skew & Kurt & SW & JB & L1O & L1A \\
\hline
$Z_V$ & 96 & 0.3644 & 0.590 & 0.057 & 0.009 & 0.061 & 0.401 & 0.237\\
$Z_G$ & 97 & 0.0218 & 0.614 & 2.903 & 0.000 & 0.000 & 0.474 & 0.253\\
$U$  & 97 & 0.1767 & -0.816 & 1.102 & 0.003 & 0.000 & 0.291 & 0.608\\
\hline
\end{tabular}
\caption{Analysis of regression residuals. Skew = skewness, Kurt = kurtosis, SW and JB =  Shapiro-Wilk and Jarque-Bera $p$-values, L1O and L1A = L1 norm for the first 5 lags of the ACF for $Z$ and $|Z|$, Length = the size of $Z$.}
\label{table:residuals-bubble}
\end{table}

We consider modeling $G/V$ as IID. But we also add the intercept, similarly to switching from~\eqref{eq:new-linear} to~\eqref{eq:vol-linear}. We get:
\begin{equation}
\label{eq:earnings-vol}
\frac{G(t)}{V(t)} = g + \frac{c_G}{V(t)} + Z_G(t) \Leftrightarrow G(t) = gV(t) + c_G + V(t)Z_G(t).
\end{equation}
Combning~\eqref{eq:trailing},~\eqref{eq:bubble-def}, and~\eqref{eq:bubble-ar}, we get the model:

\begin{align}
\label{eq:bubble-model}
\begin{split}
\overline{E}(t) &= \frac1L\sum\limits_{s=0}^{L-1}E(t-s);\\
G(t) &= \ln E(t) - \ln E(t-1),\\
\overline{G}(t) &= \ln\overline{E}(t) - \ln\overline{E}(t-1);\\
H(t) &= \sum\limits_{s=1}^tQ(s)  - \ln\overline{E}(t) - ct;\\
H(t) - h &= b(H(t-1) - h) + U(t);\\
\ln V(t) &= \alpha_V  + \beta_V\ln V(t-1) + Z_V(t);\\
\end{split}
\end{align}

We have $c_G = 0.123275, g = -0.005788$. The $p$-values for the Student $T$-test are $0.8\%$ and $32\%$, respectively. The correlation matrix of residuals is 
$$
\mathbf{\Sigma}_3 = 
\begin{bmatrix}
1 & -0.18 & -0.42\\
-0.18 & 1 & 0.1\\
-0.42 & 0.1 & 1\\
\end{bmatrix}
$$
All these computations in this section are done in Python code \texttt{bubble.py}

\begin{theorem} If $\beta_V, b \in (0, 1)$ and $(Z_G(t), Z_V(t), U(t))$ are IID with mean zero, then the system~\eqref{eq:bubble-model} has a stationary version. 
\end{theorem}

\begin{proof}
Stationarity of $V$ and $H$ is immediate. Clearly, $G$ is also stationary. Finally, let us show stationarity of $Q$. We can rewrite~\eqref{eq:bubble-def} and~\eqref{eq:bubble-ar} as
\begin{align*}
Q(t) &= H(t) - H(t-1) + c + \overline{G}(t);\\
\overline{G}(t) &= \ln\overline{E}(t) - \ln\overline{E}(t-1).\\
\end{align*}
It suffices to prove $\overline{G}(t)$ is stationary. But we can rewrite
\begin{align*}
\overline{G}(t) &= \ln\frac{\overline{E}(t)}{E(t-L)} - \ln\frac{\overline{E}(t-1)}{E(t-L)}  = \ln\sum\limits_{s=0}^{L-1}\frac{E(t-s)}{E(t-L)} - \ln\sum\limits_{s=0}^{L-1}\frac{E(t-s-1)}{E(t-L)}.
\end{align*}
The terms $E(t-s)/E(t-L)$ are stationary for each $s$, and are jointly stationary for $s = 0, \ldots, L$, for the simple reason that these  all are products of $G(t), \ldots, G(t-L+1)$. For example, $E(t)/E(t-L) = G(t)\cdot\ldots\cdot G(t-L+1)$. This proves stationarity for $\overline{G}(t)$, and therefore for $Q(t)$. 
\end{proof}

\subsection{Lag selection} Try other averaging window $L$ from 1 to 10: The Python code used in this section is \texttt{selection.py}
\begin{equation}
\label{eq:ar-bubble}
H(t) - h = b(H(t-1) - h) + U(t).
\end{equation}
We analyze residuals $U$; results are in~\textsc{Table}~\ref{table:selection}. We conclude that all cases allow us to model $U$ as IID, so~\eqref{eq:ar-bubble} is a true autoregression. However,  we cannot model $U$ as Gaussian, except for the cases $L = 2$ and $L = 3$. 

We chose $L = 10$ to make our research match Shiller CAPE ratio. We think that low $L$ would make our valuation measure fluctuate wildly in the short run, together with annual earnings. We do not need to simulate the residuals $U$ from~\eqref{eq:ar-bubble} in our simulations, anyway. 

\begin{table}[h]
\centering
\begin{tabular}{|cccccccc|}
\hline
$L$ & $R^2$ & SW & JB & Skew & Kurt & L1O & L1A \\
\hline
1 & 17\% & 2.8\% & 0.02\% & 0.55 & 1.7 & 0.062 & 0.604 \\
2 & 12\% & 3.1\% & 25\% & -0.4 & 0.18 & 0.39 & 0.43 \\
3 & 11\% & 20\% & 12\% & -0.38 & 0.71 & 0.41 & 0.378\\
4 & 10\% & 5.9\% & 1.7\% & -0.31 & 1.27 & 0.392 & 0.546\\
5 & 10\% & 9.1\% & 3.1\% & -0.52 & 0.78 & 0.446 & 0.946\\
6 & 10\% & 11\% & 5.2\% & -0.48 & 0.74 & 0.438 & 0.753\\
7 & 10\% & 13\% & 3.5\% & -0.56 & 0.65 & 0.436 & 0.725\\
8 & 9\% & 10\% & 7.5\% & -0.53 & 0.39 & 0.302 & 0.761\\
9 & 9\% & 2.6\% & 0.8\% & -0.67 & 0.73 & 0.277 & 0.632\\
10 & 10\% & 0.3\% & 0.04\% & -0.82 & 1.1 & 0.291 & 0.608\\
\hline
\end{tabular}
\bigskip
\caption{Residuals}
\label{table:selection}
\end{table}

In addition, we might include volatility in~\eqref{eq:ar-bubble}: $H(t) - h = b(H(t-1) - h) + V(t)U(t)$. However, analysis of the ACF for the residuals $U(t)$ shows it works only for short averaging window. To keep it simple, we will focus in this article on the model~\eqref{eq:ar-bubble}, delegating the other equation to further research.

\section{The Main Model}

\subsection{Overview} We now combine three parts in the Python code \texttt{main.py}

\begin{itemize}
\item the simplified model~\eqref{eq:duration-system} of three asset classes from Section 4, with only volatility $V$ and rate $R$ as factors
\item the combined model of S\&P earnings and returns~\eqref{eq:bubble-model} from Section 5, without international stocks and corporate bonds
\item the term spread $S$ as factor for returns and earnings growth
\end{itemize}

\subsection{Term spread} Since the long-short term spread $S(t)$ can be both positive and negative, we apply the autoregression of order 1 without taking the logarithm:
\begin{equation}
\label{eq:spread}
S(t) = a_S + b_SS(t-1) + Z_S(t).
\end{equation}
We estimate $a_S = 0.6437$ and $b_S = 1 - 0.4605$. The $p$-values for Student $T$-test for $a_S = 0$ or for $b_S = 1$ are less than $0.1\%$. 

\subsection{Earnings growth} We model earnings growth $G(t)$ as linear regression with respect to change in rate $R$, and term spread $S$:
\begin{equation}
\label{eq:growth-reg}
G(t) = c_G - d_G(R(t) - R(t-1)) - b_GS(t-1) + Z_G(t).
\end{equation}
But remember that we divide $G(t)$ by $V(t)$ to make these growth terms closer to Gaussian. Similarly to the switch from~\eqref{eq:simple-linear} to~\eqref{eq:vol-linear}, we need to change~\eqref{eq:growth-reg} to
\begin{equation}
\label{eq:growth-reg-vol}
G(t) = a_G + b_GS(t-1) + c_GV(t) - d_G(R(t) - R(t-1)) + V(t)Z_G(t). 
\end{equation}

\begin{table}[h]
\centering
\begin{tabular}{lrrrrrr}
\hline
\textbf{variable} & \textbf{coefficient}  & \textbf{estimate} & \textbf{stderr} & \textbf{t} & \textbf{P$>\vert$t$\vert$} & \textbf{[0.025, 0.975]} \\
\hline
constant & $a_G$ & 0.07756  & 0.048 & 1.630  & 0.107 & [-0.017, 0.172] \\
term spread & $b_G$ & 0.04786  & 0.018 & 2.667  & 0.009 & [0.012, 0.083] \\
volatility & $c_G$   & 0.007841 & 0.006 & -1.370 & 0.174 & [-0.019, 0.004] \\
rate change  & -$d_G$ & 0.03721 & 0.025 & 1.466  & 0.146 & [-0.013, 0.088] \\
\hline
\end{tabular}
\caption{Earnings growth regression results for~\eqref{eq:growth-reg-vol}. Here $R^2 = 13.8\%$. Note $c_G$ and $d_G$ are not significant, we could remove them. But we chose to retain them, to stress dependence upon the BAA rate and term spread.}
\label{table:regression-earnings}
\end{table}

Also, we model S\&P returns $Q(t)$ as a linear regression versus change in rate $R(t) - R(t-1)$, term spread $S(t-1)$, and the valuation measure $H(t-1)$. It is important to divide by the annual volatility, as in~\eqref{eq:vol-linear} or~\eqref{eq:growth-reg-vol}:
\begin{equation}
\label{eq:all-returns}
Q(t) = a_Q - k_QH(t-1) - d_Q(R(t) - R(t-1)) - b_QS(t-1) + c_QV(t) + V(t)Z_Q(t).
\end{equation}

\begin{table}[h]
\centering
\begin{tabular}{lrrrrrr}
\hline
\textbf{variable} & \textbf{coefficient} & \textbf{estimate} & \textbf{stderr} & \textbf{t} & \textbf{P$>\vert$t$\vert$} & \textbf{[0.025, 0.975]} \\
\hline
constant & $a_Q$  & 0.26851  & 0.032 & 8.440  & 0.000 & [0.205, 0.332] \\
term spread & -$b_Q$ & -0.03412 & 0.012 & -2.913 & 0.004 & [-0.057, -0.011] \\
volatility & $c_Q$ & -0.013568 & 0.004 & -3.523 & 0.001 & [-0.021, -0.006] \\
rate change & -$d_Q$ & -0.078238 & 0.017 & -4.738 & 0.000 & [-0.111, -0.045] \\
valuation measure & -$k_Q$ & -0.1644 & 0.045 & -3.667 & 0.000 & [-0.253, -0.075] \\
\hline
\end{tabular}
\caption{Domestic returns regression results~\eqref{eq:all-returns}; $R^2 = 53.1\%$. All coefficients are significantly different from zero. Residuals are IID Gaussian.}
\label{tab:regression_results_model3}
\end{table}

Same is done for international returns: Although these are non-US stocks, we study their dependence upon US financial factors: The valuation measure $H$, the change in BAA rate, and the term spread. 
\begin{equation}
\label{eq:intl-returns}
I(t) = a_I - k_IH(t-1) - d_I(R(t) - R(t-1)) - b_IS(t-1) + c_IV(t) + V(t)Z_I(t).
\end{equation}
Dependence is much weaker, but sometimes still significant. See results in \textsc{Table}~\ref{table:intl-regression-results}.

\begin{table}[h]
\centering
\begin{tabular}{lrrrrrr}
\hline
\textbf{variable} & \textbf{coefficient} & \textbf{estimate} & \textbf{stderr} & \textbf{t} & \textbf{P$>\vert$t$\vert$} & \textbf{[0.025, 0.975]} \\
\hline
constant  & $a_I$    & 0.2627 & 0.077 & 3.425  & 0.001 & [0.109, 0.417] \\
term spread  & -$b_I$   & 0.0030  & 0.019 & 0.160  & 0.874 & [-0.035, 0.041] \\
volatility     & $c_I$    & -0.0186 & 0.008 & -2.336 & 0.024 & [-0.035, -0.003] \\
rate change & -$d_I$  & -0.0494 & 0.025 & -2.006 & 0.050 & [-0.099, 0.000055] \\
valuation measure  & -$k_I$ & -0.0955 & 0.083 & -1.155 & 0.254 & [-0.262, 0.071] \\
\hline
\end{tabular}
\caption{International returns regression results~\eqref{eq:intl-returns}; $R^2 = 41.5\%$. Coefficients $b_I$ and $k_I$ are not significantly different from zero.}
\label{table:intl-regression-results}
\end{table}

We consider a version of~\eqref{eq:intl-returns} without the spread variable: $b_I = 0$. We see that $R^2$ is not improved. We use this simplified version for the financial simulator, see \textsc{Table}~\ref{table:intl-regression-results-cut}.

\begin{table}[h]
\centering
\begin{tabular}{lrrrrrrr}
\hline
\textbf{variable} & \textbf{coefficient} & \textbf{estimate} & \textbf{stderr} & \textbf{t} & \textbf{P$>|\text{t}|$} & \textbf{[0.025, 0.975]} \\
\hline
constant   &  $a_I$  & 0.2689 & 0.065 & 4.108 & 0.000 & [0.137, 0.400] \\
volatility     & $c_I$     & -0.0188 & 0.008 & -2.412 & 0.019 & [-0.034, -0.003] \\
rate change & -$d_I$ & -0.0514 & 0.021 & -2.447 & 0.018 & [-0.094, -0.009] \\
valuation measure    & -$k_I$ & -0.0941 & 0.081 & -1.156 & 0.253 & [-0.258, 0.069] \\
\hline
\end{tabular}
\caption{International returns regression results~\eqref{eq:intl-returns} without the spread variable; $R^2 = 41.5\%$, same as for the general version of this regression. Residuals are IID Gaussian.}
\label{table:intl-regression-results-cut}
\end{table}

\subsection{Combined equations for the complete model} Now let us combine all regression equations for returns $Q, I, B$ of three asset classes: Domestic stocks~\eqref{eq:all-returns}, international stocks~\eqref{eq:intl-returns}, and domestic bonds~\eqref{eq:B}. We get the following system~\eqref{eq:complete-returns}:

\begin{align}
\label{eq:complete-returns}
\begin{split}
Q(t) &= a_Q + b_QV(t) - c_QS(t-1) - d_Q(R(t) - R(t-1)) - k_QH(t-1) + V(t)Z_Q(t);\\
I(t) &= a_I + b_IV(t) - c_IS(t-1) - d_I(R(t) - R(t-1)) - k_IH(t-1) + V(t)Z_I(t);\\
B(t) &= 0.01R(t-1) - a_B - d_B(R(t) - R(t-1)) + Z_B(t).
\end{split}
\end{align}

Also, combine evolution equations~\eqref{eq:vol},~\eqref{eq:baa},~\eqref{eq:spread},~\eqref{eq:growth-reg-vol},~\eqref{eq:bubble-model} (except the first one) of the four factors in~\eqref{eq:complete-factors}: volatility $V$, BAA rate $R$, term spread $S$, and the valuation measure $H$.
\begin{align}
\label{eq:complete-factors}
\begin{split}
\ln V(t) &= a_V + b_V\ln V(t-1) + Z_V(t);\\
\ln R(t) &= a_R + b_R\ln R(t-1) + Z_R(t);\\
S(t) &= a_S + b_SS(t-1) + Z_S(t);\\
G(t) &= a_G + b_GV(t) - c_GS(t-1) - d_G(R(t) - R(t-1)) + V(t)Z_G(t);\\
H(t) &= Q(t) - \overline{G}(t) + c + H(t-1).
\end{split}
\end{align}

In~\eqref{eq:complete-factors}, we have one nonrandom (deterministic) equation for $H$. Each another equation contains a series of IID residuals. After we fit the autoregression (with trend) for $H$, we do not use the innovations $U$ of that autoregression. This stands in contrast to~\eqref{eq:complete-returns} where all three equations have random residuals. Thus the entire simulation model~\eqref{eq:complete-returns} and~\eqref{eq:complete-factors} contain only 7 and not 8 series of innovations. 

Also, in~\eqref{eq:complete-factors} we define $G(t) = \ln E(t) - \ln E(t-1)$ and $\overline{G}(t) = \ln\overline{E}(t) - \ln\overline{E}(t-1)$ are growth rates for annual earnings and 10-year trailing averaged earnings. Note that $\overline{G}(t)$ is a function of $G(t), E(t-1), \ldots, E(t-10)$. 

Each residuals series is IID. But some of them are not Gaussian. Thus it is reasonable to model these residuals as IID. By construction, they have mean zero. Below, we present \textsc{Table}~\ref{table:residuals} of statistical functions for each series of these residuals. We include the skewness, kurtosis (both being zero for the normal distribution), $p$-values for Shapiro-Wilk and Jarque-Bera normality tests, and the sum of absolute values (the $L_1$ norm) of the first 5 lags for ACF of original and absolute values of residuals. We see that all residuals for factors (volatiltiy, spread, rate, growth) are IID but not Gaussian. However, residuals for all three returns are IID Gaussian. The main goal is accomplished: We found a time series model where innovations are what they are supposed to be: IID. However, we need to use some other method (not a multivariate Gaussian distribution) to model them. 

\begin{table}[ht]
\centering
\begin{tabular}{|c|cccccccc|}
\hline
Reg & Size & Stdev & Skew & Kurt & SW & JB & L1O & L1A \\
\hline
$Z_V$ & 96 &  0.3644 & 0.590 & 0.057 & 0.009 & 0.061 & 0.401 & 0.237 \\
$Z_S$ & 97 & 1.0537 & 0.475 & 2.150 & 0.008 & 0.000 & 0.418 & 0.689 \\
$Z_R$ & 97 & 0.1357 & 0.814 & 1.409 & 0.007 & 0.000 & 0.177 & 0.360 \\
$Z_G$ & 97 & 0.021 & 0.614 & 2.903 & 0.000 & 0.000 & 0.474 & 0.253 \\
$Z_Q$ & 97 & 0.0135 & 0.039 & 0.157 & 0.344 & 0.940 & 0.413 & 0.590 \\
$U$ & 97 & 0.1767 & -0.816 & 1.102 & 0.003 & 0.000 & 0.291 & 0.608 \\
$Z_I$ & 55 & 0.0181 & 0.073 & -0.512 & 0.621 & 0.723 & 0.429 & 0.635 \\
$Z_B$ & 52 & 0.0263 & 0.193 & 0.238 & 0.857 & 0.800 & 0.878 & 0.706 \\
\hline
\end{tabular}
\caption{Analysis of regression residuals. Skew = skewness, Kurt = kurtosis, SW and JB =  Shapiro-Wilk and Jarque-Bera $p$-values, L1O and L1A = L1 norm for the first 5 lags of the ACF for $Z$ and $|Z|$, Length = the size of $Z$}
\label{table:residuals}
\end{table}

\begin{table}[h!]
\centering
\begin{tabular}{|l|rrrrrrrr|}
\hline
$Z$ & $Z_V$ & $Z_S$ & $Z_R$ & $Z_G$ & $U$ & $Z_Q$ & $Z_I$ & $Z_B$ \\
\hline
$Z_V$ & 1.0000  & 0.0538  & 0.2978  & -0.1800 & -0.4234 & -0.0516 & -0.0753 & 0.0725 \\
$Z_S$ & 0.0538  & 1.0000  & -0.1817 & -0.0661 & -0.0608 & -0.1581 & -0.0397 & -0.0951 \\
$Z_R$ & 0.2978  & -0.1817 & 1.0000  & -0.1628 & -0.5177 & -0.0979 & -0.0265 & -0.1470 \\
$Z_G$ & -0.1800 & -0.0661 & -0.1628 & 1.0000  & 0.1524  & 0.1204  & 0.0724  & -0.1878 \\
$U$ & -0.4234 & -0.0608 & -0.5177 & 0.1524  & 1.0000  & 0.7038  & 0.3228  & 0.2714 \\
$Z_Q$ & -0.0516 & -0.1581 & -0.0979 & 0.1204  & 0.7038  & 1.0000  & 0.2768  & 0.2817 \\
$Z_I$ & -0.0753 & -0.0397 & -0.0265 & 0.0724  & 0.3228  & 0.2768  & 1.0000  & -0.0948 \\
$Z_B$  & 0.0725  & -0.0951 & -0.1470 & -0.1878 & 0.2714  & 0.2817  & -0.0948 & 1.0000 \\
\hline
\end{tabular}
\caption{Correlation matrix of series of residuals, see \texttt{corrMatrix.xlsx}}
\end{table}

\subsection{Main stability results} 

Consider the complete system of factors and asset returns
\begin{equation}
\label{eq:extended-system}
(V, R, S, G, \overline{G}, H, Q, I, B),
\end{equation}
evolving as~\eqref{eq:complete-returns} and~\eqref{eq:complete-factors}. The system~\eqref{eq:extended-system} is not Markov. Indeed, the averaged earnings growth $\overline{G}$ depends on the last 11 years of earnings, not the last two years. But we can fix the last 10 years of earnings. Under this condition, the system~\eqref{eq:extended-system} is Markov: Its value at step $t$ depends only on the value at step $t-1$. Of course, this Markov process is not time homogeneous, since the equations use these initial 10 years of earnings anyway. 

\begin{theorem}
\label{thm:new}
Assume the series $(Z_Q(t), Z_I(t), Z_B(t), Z_G(t), Z_R(t), Z_S(t), Z_V(t))$ in $\mathbb R^7$ are IID with mean zero. Assume  $k_Q \in (0, 2)$ and $\beta_R, \beta_S, \beta_V \in (-1, 1)$. Then the extended system~\eqref{eq:extended-system} has a stationary version. 
\label{thm:main}
\end{theorem}

\subsection{Proof of the main result} If we prove stationarity of other components, then the stationarity of~\eqref{eq:B} and~\eqref{eq:intl-returns} are immediate. Stationarity for $V$, $R$, $S$ follows from standard properties of autoregressions, stated in Lemma~\ref{lemma:AR} below. Now rewrite the equations~\eqref{eq:complete-factors}:
\begin{align}
\begin{split}
H(t) &= H(t-1)(1 - k_Q) + F(t),\\
F(t) &:= a_Q - c - d_Q(R(t) - R(t-1)) - b_QS(t-1) + c_QV(t) + V(t)Z_Q(t) - \overline{G}(t).
\end{split}
\end{align}
Define the function $\varphi(x) := (\ln x)_+ := \max(\ln x, 0)$. It is clear that 
\begin{equation}
\label{eq:X-log}
\mathbb E|X| < \infty \Rightarrow \mathbb E\varphi(|X|) < \infty. 
\end{equation}
Apply Lemma~\ref{lemma:AR}. From~\eqref{eq:X-log}, for each of 5 series of residuals $Z_Q(t), Z_G(t), Z_R(t), Z_S(t), Z_V(t)$, 
\begin{equation}
\label{eq:log-residuals}
\mathbb E\varphi(|Z_A(t)|) < \infty,\quad A = G, Q, R, S, V.
\end{equation}

Following the main results of \cite{Brandt0, Brandt}, we need to show two results:
$\ln|1 - k_Q| < 0$ (which trivially follows from $0 < k_Q < 2$) and $\mathbb E\varphi(|F(t)|) < \infty$. Using Lemma~\ref{lemma:sum-product}, we reduce it to:
\begin{equation}
\label{eq:ar-varphi}
\mathbb E\varphi(R(t)) < \infty,\quad \mathbb E\varphi(R(t-1)) < \infty; \quad \mathbb E\varphi(V(t)) < \infty;
\end{equation}
\begin{equation}
\label{eq:ar-s}
\mathbb E\varphi(S(t-1)) < \infty; 
\end{equation}
\begin{equation}
\label{eq:G}
\mathbb E\varphi(\overline{G}(t)) < \infty.
\end{equation} 
The statements~\eqref{eq:ar-varphi} follow from Lemma~\ref{lemma:AR}. The statement~\eqref{eq:ar-s} also does, using~\eqref{eq:X-log}. Finally, let us show~\eqref{eq:G}. It is easy to see that for positive numbers $x_1, \ldots, x_n, y_1, \ldots, y_n$, we have:
$$
 \min\left(\frac{x_1}{y_1}, \ldots, \frac{x_n}{y_n}\right) \le \frac{x_1 + \ldots + x_n}{y_1 + \ldots + y_n} \le \max\left(\frac{x_1}{y_1}, \ldots, \frac{x_n}{y_n}\right).
$$
Applying the logarithm to both sides, we get:
\begin{equation}
\label{eq:min-max}
\min\left(\ln\frac{x_1}{y_1}, \ldots, \ln\frac{x_n}{y_n}\right) \le \ln\frac{x_1 + \ldots + x_n}{y_1 + \ldots + y_n} \le \max\left(\ln\frac{x_1}{y_1}, \ldots, \ln\frac{x_n}{y_n}\right).
\end{equation}
However, for real numbers $z_1, \ldots, z_n$, we have:
$$
-\sum_{i=1}^n|z_i| \le \min\limits_{i = 1, \ldots, n}(z_i) \le \max\limits_{i = 1, \ldots, n}(z_i) \le \sum_{i=1}^n|z_i|.
$$
Applying this to~\eqref{eq:min-max}, we have:
\begin{equation}
\label{eq:x-y}
-\sum_{i=1}^n\left|\ln\frac{x_i}{y_i}\right| \le \ln\frac{x_1+\ldots +x_n}{y_1+\ldots+y_n} \le \sum_{i=1}^n\left|\ln\frac{x_i}{y_i}\right|.
\end{equation}
Finally, let $n = L$, and let $x_k = E(t-k)$ and $y_k = E(t-k-1)$ for $k = 1, \ldots, L$ in~\eqref{eq:x-y}. Using~\eqref{eq:bubble-model}, we get almost surely
\begin{equation}
\label{eq:AS}
|\overline{G}(t)| \le \frac1L\sum\limits_{k=0}^{L-1}|G(t-k)|.
\end{equation}
To show~\eqref{eq:G}, we need to apply Lemma~\ref{lemma:sum-product} and~\eqref{eq:AS}, and finally to show 
$$
\mathbb E\varphi(|G(t)|) < \infty.
$$
But this, in turn, follows from Lemma~\ref{lemma:sum-product} and~\eqref{eq:X-log}.  Thus the proof of ~\eqref{eq:G}, and with it the main result, is complete.

\begin{lemma} If two random variables $X, Y$ (not necessarily independent) have 
\begin{equation}
\label{eq:finite-phi}
\mathbb E\varphi(|X|) < \infty,\quad \mathbb E\varphi(|Y|) < \infty,
\end{equation}
then the sum and product of these two random variables also satisfy these properties:
$$
\mathbb E\varphi(|X + Y|) < \infty,\quad \mathbb E\varphi(|XY|) < \infty.
$$
\label{lemma:sum-product}
\end{lemma}

\begin{proof} First, let us show this for $X + Y$. This function $\varphi$ is nondecreasing but not concave:
$$
\varphi(x) = 
\begin{cases} 
\ln x, \, x \ge 1;\\ 
0,\, x \in [0, 1]. 
\end{cases}
$$
But it has a {\it concave hull:} the smallest concave $\psi : \mathbb R_+ \to \mathbb R$ such that $\psi(x) \ge \varphi(x)$ for all $x \ge 0$. We can find it by drawing a tangent line from $(0, 0)$ to $x \mapsto \psi(x) = \ln x$ at the point $(a, \ln a)$ for $a > 1$. We compare the slope of this line: $\ln a/a = \psi'(a) = 1/a$, equivalent to  $\ln a = 1$ and $a = e$. Therefore,
$$
\varphi(x) \le \psi(x) = 
\begin{cases}
x/e,\, x \le e;\\
\ln x,\, x \ge e.
\end{cases}
$$
This function $\psi$ satisfies 
\begin{equation}
\label{eq:psi-linearity}
\psi(x + y) \le \psi(x) + \psi(y),\quad x, y \ge 0.
\end{equation}
Indeed, from concavity of $\psi$, we have: $\psi(x + y) - \psi(x) \le \psi(y) - \psi(0) = \psi(y)$. Thus for any random variable $X$, 
\begin{equation}
\label{eq:phi-psi}
\mathbb E\varphi(|X|) < \infty \Leftrightarrow \mathbb E\psi(|X|) < \infty.
\end{equation}
Applying~\eqref{eq:phi-psi} to~\eqref{eq:finite-phi}, we get: 
\begin{equation}
\label{eq:finite-psi}
\mathbb E\psi(|X|) < \infty,\quad \mathbb E\psi(|Y|) < \infty.
\end{equation}
Applying~\eqref{eq:psi-linearity} to~\eqref{eq:finite-psi}, we get: $\mathbb E\psi(|X + Y|) < \infty$. Finally, apply~\eqref{eq:phi-psi} again and complete the proof that~\eqref{eq:finite-phi} implies $\mathbb E\varphi(|X + Y|) < \infty$. Next, we prove this for $XY$. Indeed, 
\begin{equation}
\label{eq:varphi-product}
\varphi(xy) \le \varphi(x) + \varphi(y),\quad x, y \ge 0.
\end{equation}
To prove~\eqref{eq:varphi-product}, we need to consider four cases, comparing $x, y, xy$ with $1$: 
\begin{itemize}
\item if $x, y \le 1$, then $xy \le 1$, so $\varphi(x) = \varphi(y) = \varphi(xy) = 0$;
\item if $x, y \ge 1$, then $xy \ge 1$, so $\varphi(x) = \ln x$, $\varphi(y) = \ln y$, and $\varphi(xy) = \ln(xy) = \ln x + \ln y$;
\item if $x \le 1$ and $y \ge 1$ so that $xy \le 1$, then $\varphi(x) = 0$ and $\varphi(y) = \ln y$, but 
$$
\varphi(xy) = 0 \le 0 + \ln y = \varphi(x) + \varphi(y); 
$$
\item if $x \le 1$ and $y \ge 1$ so that $xy \ge 1$, then $\varphi(x) = 0$ and $\varphi(y) = \ln y$, but 
$$
\varphi(xy) = \ln(xy) = \ln x + \ln y \le 0 + \ln y = \varphi(x) + \varphi(y).
$$
\end{itemize}
This completes the proof of~\eqref{eq:varphi-product}, and with it the lemma.
\end{proof}

\begin{lemma} Take IID $\delta(t)$ for $t = 0, 1, 2, \ldots$ with finite mean. Fix a $c \in (0, 1)$. Then there exists a unique stationary version of the autoregression $Y(t) = cY(t-1) + \delta(t)$, with $\mathbb E|Y(t)| < \infty$. 
\label{lemma:AR}
\end{lemma}

\begin{proof} This is a well-known result, see the classic book \cite{Series}: Extending $\delta(t)$ for negative $t$, we solve
$$
Y(t) = \sum_{k = 0}^{\infty}c^k\delta(t-k)
$$
and observe that this series converges in $L^1$. Indeed, the $L^1$ norm 
$|\!|\delta|\!|_1 := \mathbb E|\delta(s)|$ is finite and the same for all $s$, so 
$$
\mathbb E|Y(t)| \le \sum\limits_{k=0}^{\infty}c^k\cdot\mathbb E|\delta(t-k)| = |\!|\delta|\!|_1\cdot\sum\limits_{k=0}^{\infty}c^k = \frac{|\!|\delta|\!|_1}{1-c}.
$$
This completes the proof of Lemma~\ref{lemma:AR}.
\end{proof}

\section{Financial Simulator} 

\subsection{User interface} We made an online financial app \texttt{asarantsev.pythonanywhere.com} implementing annual simulation of returns of these three asset classes: USA stocks, international stocks, and USA corporate bonds. We use four factors: the valuation measure $H$, the BAA rate $R$, the term spread $S$, and the volatility $V$. The back end is written in Python with Flask framework \texttt{flask\_app.py} and the front end is written in HTML in \texttt{main\_page.html} In an advanced version of this simulator, we allow to change initial values of these factors. The front end for this option is written in HTML in \texttt{complete\_page.html} The residuals are written in Excel file \texttt{innovations.xlsx}

We fix the last 10 years of earnings, 2015--2024. We consider only nominal data: We failed to model inflation-adjusted bond returns with meaningful regressions and IID residuals. 

We allow time horizon of 1--50 years. Almost none plans for more than that, and visualizing the simulation then becomes difficult. The stock/bond proportion can be adjusted at the start and at the end (for example, 80/20 now and 60/40 at the end of the simulation, which is 30 years from now). In between, it is interpolated linearly. Within the stock portfolio, the proportion of domestic and international stocks is set by the user and is constant. 

The number of simulations is 10000. We rank the simulations by final wealth, and pick the 5 paths which correspond to 10\%, 30\%, 50\% (median), 70\%, and 90\%. We picked them because we consider these to be representative of the entire distribution of possible paths. It is remarkable to see how much variance is in these returns. The 90\%-path final wealth can exceed the 10\%-path final wealth by factors of 100, if the time horizon is long enough and the stock market proportion in the portfolio is more than 50\%. 
The output is in \textsc{Figure}~\ref{fig:sim}.

\begin{figure}[t]
\includegraphics[width = 15cm]{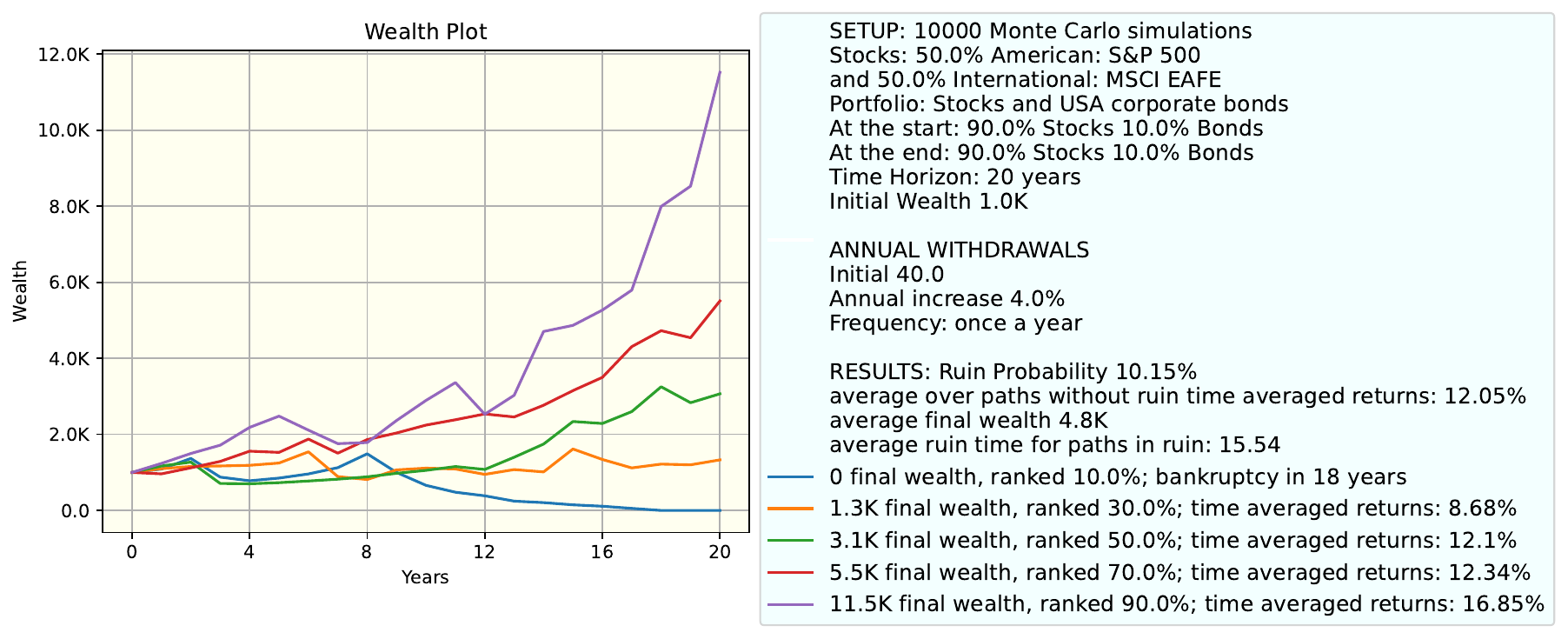}
\caption{Simulator Output. On the left, we have the graphs of 10\%, 30\%, 50\%, 70\%, 90\%-ranked paths by final wealth. On the right, we have the written problem input, and output summary statistics.}
\label{fig:sim}
\end{figure}

\subsection{White noise simulation}  We simulate the $\mathbb R^d$-valued white noise ($d = 7$) using multivariate kernel density estimation (KDE) with Gaussian kernel. Recall that we have $N = 2024 - 1927 = 97$ data points $z(1), \ldots, z(N)$. For the bandwidth, we use the Silverman's rule of thumb, slightly modified: The covariance matrix of this multivariate Gaussian kernel is diagonal, so de facto it is a product of one-dimensional Gaussian kernels. Now, take each component. Let $s$ be its empirical standard deviation, and let $r$ be its {\it interquartile range:} the difference between the 75\% and the 25\% quantiles. Then the bandwidth is 
\begin{equation}
\label{eq:sigma-KDE}
\sigma = \left(\frac{4}{d+2}\right)^{1/(d + 4)} \cdot N^{-1/(d + 4)}\cdot \min\left(s, \frac{r}{1.34}\right).
\end{equation}
We compute $\sigma = \sigma_i$ from~\eqref{eq:sigma-KDE} for each $i$th component. We adjusted the mulivariate Silverman's rule of thumb to make this bandwidth look more like the univariate one. Thus the density looks like this:
$$
F(\mathbf{x}) = \frac1N\sum\limits_{k=1}^N\prod\limits_{i=1}^d\varphi(x_i, z_i(k), \sigma_i),
$$
where $\varphi(y, m, s)$  is the probability density of $\mathcal N(m, s^2)$. 

However, some series of residuals do not have full 97 data points, see \textsc{Table}~\ref{table:residuals}. We provide the following method of filling the missing data in the Excel file \texttt{innovations.xlsx} implemented in \texttt{innovations.py} and with output written in another Excel file \texttt{filled.xlsx}

Assume $(X_0, X_1, \ldots, X_d)$ is a$(d+1)$-dimensional IID sample, such that $X_0$ does not have full data, but all other components  do have complete data. Our algorithm is:

\begin{itemize}
\item Fit the linear regression $X_0 = a_0 + a_1X_1 + \ldots + a_dX_d + Y$ for the data points when all components are present. Compute the residuals $Y$. 
\item Sample $Y'$ from existing residuals $Y$ randomly, uniformly, with replacement.
\item For the missing data $X_0$, simulate it as $X_0' = a_0 + a_1X_1 + \ldots + a_dX_d + Y'$, where $Y'$ is such sampled residual, and $X_1, \ldots, X_d$ are corresponding (existing) values. 
\end{itemize}

First, we do this for $Z_V$, with one missing data; then with $Z_I$, and then with $Z_B$. At each step, the series which were filled previously (that is, which played the role of $X_0$ above) are used as well for linear regression (becoming one of the components $X_1, \ldots, X_d$ above). 

We are aware that this is a non-standard statistical procedure. However, we could not find in the literature any other way of doing this. Classic machine learning methods such as nearest neighbors reduce variance. We understand the need for further research, and we welcome any suggestions for other existing methods. 

The multivariate kernel density estimation works as follows, using residuals written in the Excel file \texttt{filled.xlsx} 
\begin{itemize}
\item For each component $i = 1, 2, \ldots$, generate $N_i \sim \mathcal N(0, 1)$ IID.
\item Pick a random time $\tau = 1928, \ldots, 2024$ uniformly, independently of $N_1, N_2, \ldots$
\item Generate $Y = Z(\tau) + \begin{bmatrix}\sigma_1N_1 & \sigma_2N_2 & \ldots \end{bmatrix}$
\end{itemize}

\subsection{Withdrawals and contributions} Retirees or parents of college students regularly withdraw some funds from their wealth to support themselves or their children. Savers for future retirement or college regularly contribute to their (often tax-exempt) savings accounts. Endowments and souvereign wealth funds both contribute and withdraw regularly. It is easiest to model annual contributions or withdrawals, since the time step in our simulation is equal to 1 year.  If the wealth decreases to 0 or below, we fix this wealth at zero for the rest of the path. We compute the average ruin year for such paths. 

We compute the ruin probability: Which proportion of paths end up with zero wealth. Also, for each simulation, we compute which year it went bankrupt, if any. This allows us to display these ruin moments for the chosen five simulations with 10\%, 30\%, 50\%, 70\%, 90\% ranked final wealth; and the average ruin time for all paths that end in ruin. 

However, mostly contributions or withdrawals happen monthly. We allowed this without creating a separate model for simulating monthly (instead of annual) returns. Assume annual returns are simulated already. Then in each simulation, and in each time step of 1 year, we linearly interpolate the wealth process within this year. We use it to compute the resulting change in wealth for monthly withdrawals or contributions. In addition, we allow for the quarterly option (although this is less common). 

\begin{figure}[t]
\includegraphics[width = 16cm]{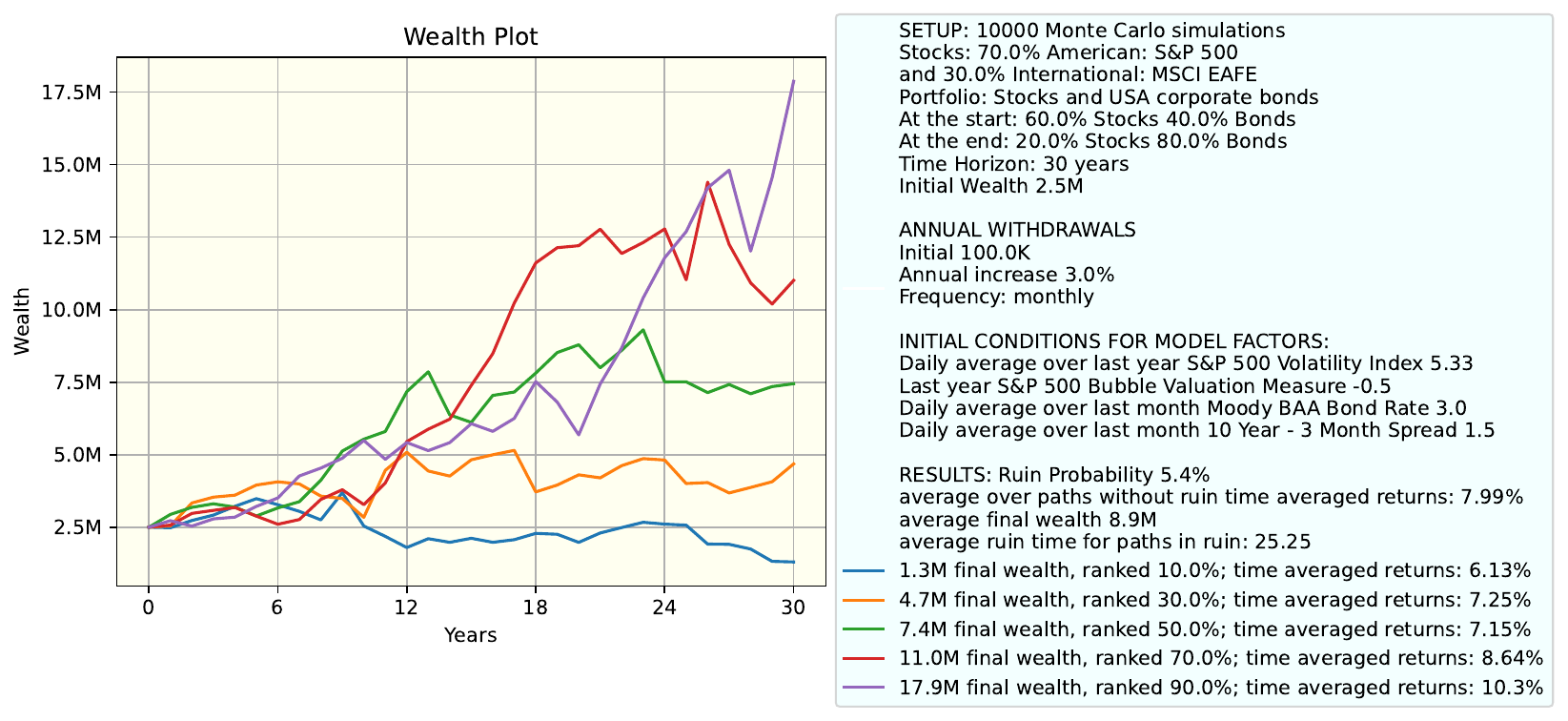}
\caption{Advanced simulator output. Initial wealth is 2.5 million, time horizon is 30 years, with initial factors: Volatility 10, BAA rate 3\%, Valuation measure -0.5, long-short spread 1.5\%. The stock/bond split is 60/40 at the start and 40/60 at the end, and the domestic/international stock split 70/30. Monthly withdrawals, annual initial amount 100K, with 3\% increase. On the left, we have the graphs of 10\%, 30\%, 50\%, 70\%, 90\%-ranked paths by final wealth. On the right, problem input, and output summary statistics.}
\label{fig:sim-advanced}
\end{figure}

\subsection{Frequent contributions/withdrawals} Assume we allow to contribute or withdraw funds from our portfolio more frequently than annually. The most common cases in practice are monthly or quarterly. We did not invent a complete model monthly or quarterly returns. 

Instead, we build on the annual model. Assume we already simulated one path, one annual step of our portfolio with arithmetic returns $r$. Assume that we do not change portfolio weights (stock/bond and domestic/international stocks) during this one year. To simulate frequent regular withdrawals, we make a very primitive assumption: Growth during this year has constant exponential rate. Take $T$ small time steps per one year: $T = 4$ for quarters, and $T = 12$ for months. Then wealth $1$ at the beginning of this year turns into wealth $(1 + r)^{t/T}$ in $t$ small time steps. Annual contributions/withdrawals of wealth $w$, done regularly at equal time intervals, is split in equal amounts $w/T$. 

Combine all this together: Contributions of wealth $1$ at small time step $t$ becomes $(1 + r)^{1 - t/T}$ at the end of this year. Thus the entire sequence of contributions becomes 
$$
\frac{w}{T}\sum\limits_{t=1}^T(1 + r)^{1 - t/T} = \frac{w}{T}\sum\limits_{s=0}^{T-1}(1 + r)^{s/T} = \frac{wr}{T}\left((1 + r)^{1/T} - 1\right)^{-1}.
$$

We see that the stock/bond allocation makes a difference for the 20-year horizon but not for the 40-year horizon, judging by the ruin probability $P$. In all cases, however, increased bond allocation decreases the variance of a portfolio. In our view, more than 30\% ruin probability is unacceptable for retirees. Since a 40-year time is a reasonable horizon for retirees, one needs to rethink or modify the 4\% withdrawal rate. We encourage the readers to try our simulator online and simulate various scenarios. 

Since we do not allow inflation adjustment (real quantities) and work only with nominal quantities, we needed to compensate for this: A reasonable annual income now would be laughably low in 50 years. We allowed annual increase or decrease (with constant rate) of contributions or withdrawals. We think that setting 4\% annual increase is a reasonable compensation for inflation. Indeed, historical annual inflation rate was around 3\%, even with the recent burst of inflation at the start of the 2020 decade.

\begin{table}
\centering
\begin{tabular}{|c|c|c|c|c|c|}
\hline
Stock/bond & Years & $P$ & $T$ & Average Wealth & 90\% Wealth \\
\hline
90/10 & 20 & 10.15\% & 15.54 & 4.8 &  11.5 \\
\hline 
90/10 & 40 & 30.66\% & 24.76 & 48.7 & 122.6 \\
\hline
60/40 & 20 & 4.35\% & 17.12 & 2.8 & 5.8 \\
\hline
60/40 & 40 & 28.2\% & 28.46 & 14.4 & 37.8 \\
\hline
40/60 & 20 & 1.51\% & 18.16 & 2.0 & 3.6\\
\hline
40/60 & 40 & 32.3\% & 31.1 & 5.2 & 14.0\\
\hline
\end{tabular}
\caption{Ruin probability $P$, ruin average time $T$, average final wealth, and final wealth corresponding to the 90\%-path.}
\label{table:critical}
\end{table}

\subsection{Retirement rules} In our previous article \cite{Valuation}, we considered validity of the classic 4\% withdrawal rate: A retired person withdraws 4\% of the initial capital the first year, and then keeps withdrawal annual amount constant after adjusting for inflation. We found this rule works well. However, there we could not simulate the earnings growth directly. As discussed in the Introduction, it is an essential part of such wealth simulation. Thus, in \cite{Valuation} we needed to make certain assumptions on earnings growth.  

In this article, we have a closed model. We can simulate this 4\% rule. As discussed above, we need to compensate for inflation by setting annual increases of withdrawals. We set 4\% annual increase in withdrawals. We assume the constant stock/bond portfolio split, and the 50/50 split within stock portfolio between domestic and international stocks. 

We see that the stock/bond allocation makes a difference for the 20-year horizon but not for the 40-year horizon, judging by the ruin probability $P$. In all cases, however, increased bond allocation decreases the variance of a portfolio. In our view, more than 30\% ruin probability is unacceptable for retirees. Since a 40-year time is a reasonable horizon for retirees, one needs to rethink or modify the 4\% withdrawal rate. We encourage the readers to try our simulator online and simulate various scenarios. 

\section{Conclusions and Further Research}

In this article, we fit several models (from simple to complicated) for wealth invested in portfolios of three asset classes: Domestic stocks, international stocks, and corporate bonds. Our goal for the model selection is:

\begin{itemize}
\item All regression residuals must be IID, and preferably Gaussian;
\item All regression coefficients must be significantly different from zero, or at least not have a large $p$-value for the Student $T$-test;
\item Regressions must have real-world financial sense.
\end{itemize}

We believe that all models considered in this article satisfy these requirements. We also provide $R^2$ values, although we do not consider this goodness-of-fit measure very important. Especially for stock returns, predictive power of various factors measured in $R^2$ is usually quite low. Stock returns are close to random walk. We cannot reasonably expect stock indices to have highly predictable returns, otherwise these would be arbitraged away. 

However, our research is quite limited. There exist many other more narrow asset classes. For the stocks, there are value/growth factors, size factors, sector portfolios (IT, medical, utilities, railroads), countries (Canada, Germany, China). We can also combine these factors: For example, German small value stocks. For background, see \cite[Chapter 6]{Textbook} or \cite[Chapter 7, Section 3]{Ang}. 

For the bonds, there are Treasury/corporate/municipal bonds, bonds of various maturities, corporate bonds with given ratings (say CCC), bonds from other countries, and their combinations. See \cite[Chapter 10, Section 2]{Ilmanen} or \cite[Chapter 9]{Ang}. 

We tried to extend our research for these asset classes, but failed: Not enough data, sometimes only available since the 1990s and 2000s; Regression residuals are not IID. 

Further research could try to find other longer data sets extending further to the past, or try nonlinear time series models to ensure residuals are IID. 

We also need to include the real (inflation-adjusted) case: As of now, we can do this for stocks but not bonds. Since any financial simulator must absolutely include bonds, we restricted ourselves to the nominal case. See also \cite[Chapter 11]{Ang}. 

What is more, further research might insert volatility into the bond market analysis, similarly to the stock market returns. In fact, for monthly bond returns, we found that including monthly average VIX improves residuals for bond regressions (makes them closer to IID and Gaussian). For the annual bond returns, this does not seem to work in our case. But we could try other bond classes, maybe for them VIX would work better. 

Last but not least, one must improve modeling of monthly/quarterly steps: Change this very primitive model described in Section 7. It seems hard to model monthly returns and monthly changes in factors, at least if you do not consider nonlinear time series. An alternative scenario seems to fix one-year step in one simulation path, and then build a model of monthly contributions/withdrawals around this step. We might need to simulate monthly returns for asset classes which we know, for limited available monthly data. 

\section{Statements and Declarations}

The authors declare no conflicts of interest. All data and code are on \texttt{GitHub/asarantsev} and the data are taken from publicly available sources. This research is not supported by any external funding. A.S. is grateful to the Department of Mathematics and Statistics for the supportive atmosphere for undergraduate research.

\end{document}